\documentclass[journal]{IEEEtran}
\IEEEoverridecommandlockouts

\usepackage{cite}
\usepackage{amsmath,amssymb,amsfonts}
\usepackage{amsthm}

\usepackage{graphicx}
\usepackage[small]{caption}

\usepackage{subcaption}
\usepackage{adjustbox}
\usepackage{textcomp}
\usepackage{multirow}
\usepackage{multicol}
\usepackage{booktabs}
\usepackage{bigstrut}
\usepackage{verbatim}
\usepackage{stackengine}
\usepackage{array}
\usepackage{rotating}
\usepackage{xcolor}
\usepackage{dblfloatfix}
\usepackage{booktabs}
\usepackage{multirow}
\usepackage{makecell}
\usepackage{adjustbox}
\usepackage{algorithm}
\usepackage{amsthm}
\newtheorem{theorem}{Theorem}

\usepackage{algpseudocode} % from algorithmicx
\usepackage[table,xcdraw]{xcolor}
\definecolor{mygreen}{RGB}{0,130,0}

\usepackage[a4paper, total={184mm,239mm}]{geometry}
\AtBeginDocument{%
  \providecommand\BibTeX{{%
    \normalfont B\kern-0.5em{\scshape i\kern-0.25em b}\kern-0.8em\TeX}}}

\begin{document}

\title{Bridging AQFP Technology Legalization and Physical Design: Layout-Aware Buffer and Splitter Insertion via Width--Depth Product Minimization}

\author{Robert S. Aviles, Ziyu Liu, Madhav Danturthi, Peter A. Beerel, \IEEEmembership{Senior Member, IEEE}\thanks{This work has been supported by ARL DEVCOM under the FSDL: ColdPhase project, grant number W911NF-24-1-0317.

The authors are with the Department of Electrical and Computer Engineering,
University of Southern California, Los Angeles, CA 90007 USA (e-mail:
rsaviles@usc.edu;  zliu4130@usc.edu, danturth@usc.edu,pabeerel@usc.edu)}}

\maketitle
 
\begin{abstract}
Adiabatic Quantum Flux Parametron (AQFP) is an emerging superconducting technology that enables ultra-low energy dissipation approaching the Shannon limit. However, its gate-level pipelining and explicit fanout constraints require technology legalization through buffer and splitter insertion to ensure path balancing and signal distribution, becoming a critical and costly step in the design flow. Prior work has focused on minimizing inserted cell count and logic depth, yet these objectives do not accurately capture the final physical design cost, which is fundamentally governed by the product of circuit width and depth.

In this article, we redefine AQFP buffer and splitter insertion optimization as minimizing the circuit width--depth product, a layout-aware metric that more accurately captures physical design area than prior cell minimization efforts. We are the first to formulate buffer and splitter insertion under this objective and prove that the resulting problem is NP-complete. To address this complexity, we develop scalable heuristics that integrates legalization with this objective.

Experimental results on standard benchmarks demonstrate that our approach achieves an average 30\% reduction in post-placement area compared to state-of-the-art methods, with only a 3\% increase in junction count, and on individual circuits up to 61\% area reduction, demonstrating the effectiveness of the proposed objective in reducing true design cost.

\end{abstract}


\section{Introduction}

The Adiabatic Quantum Flux Parametron (AQFP) superconducting logic family~\cite{AQFP} is an emerging platform for high-performance, ultra-low-power computing. AQFP circuits utilize superconducting Josephson junctions (JJs) operating at 4 K cryogenic temperatures and are particularly attractive due to their high-speed operation (up to 5 GHz), zero static power dissipation, and ultra-low-energy adiabatic operation. Even when accounting for cooling overhead, AQFP achieves up to two orders of magnitude improvement in energy–delay product (EDP) compared to CMOS implementations~\cite{cooling_overhead}. As such, it is a promising technology for quantum in-fridge computation, including control electronics~\cite{takeuchi2024microwave}, error correction ~\cite{SFQ_ECC} and data preprocessing~\cite{IcePack}. As well as classical workloads such as binary neural networks~\cite{BNN_AQFP2023}, microprocessors~\cite{MANA}, and cryptographic accelerators~\cite{SCE-NTT}.

\begin{figure}[htbp]
\includegraphics[width=\columnwidth]{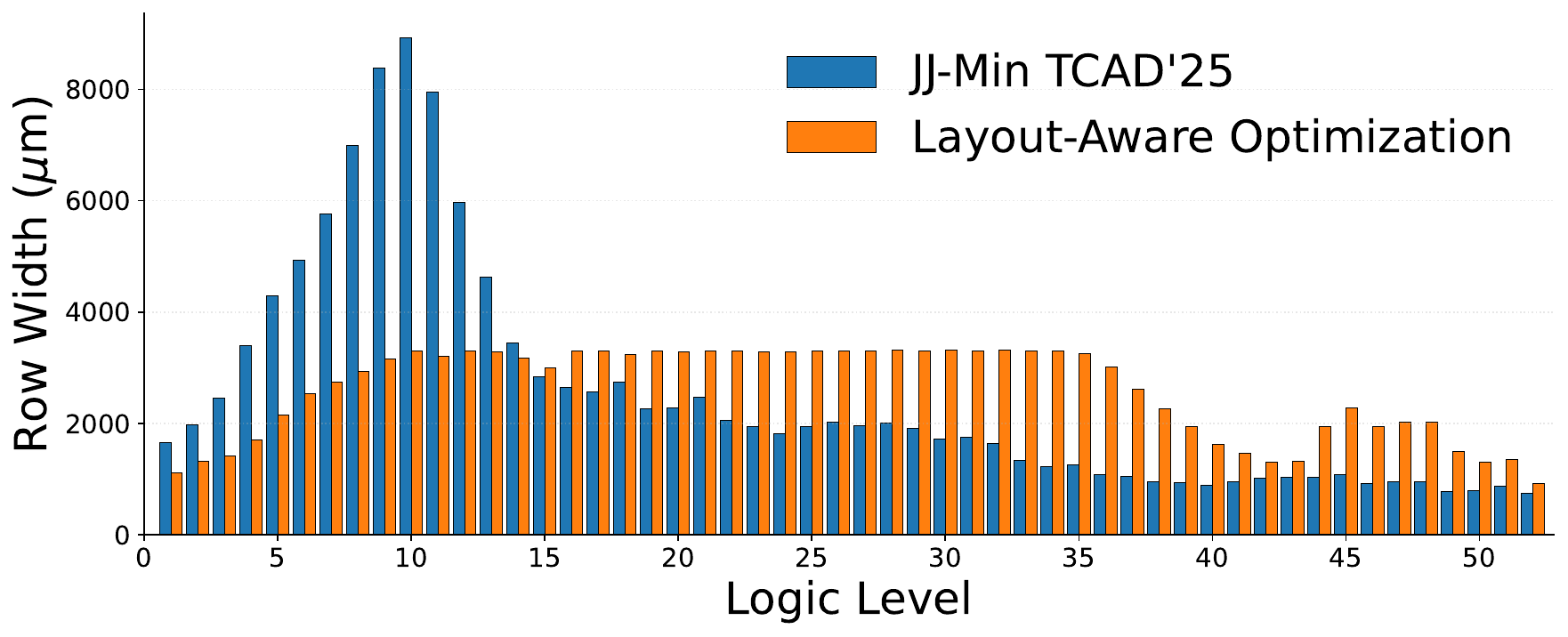}% This is a *.png file
%\vspace{-0.6cm}
\centering
\caption{Row widths of \texttt{c3540}, our layout-aware approach significantly reduces maximum row width.}\label{fig:c3540_dist}
\end{figure}

Despite these advantages, AQFP imposes unique architectural and timing constraints that complicate large-scale design and require specialized electronic design automation (EDA) support~\cite{SFQ_EDA}. In particular, all AQFP gates are clocked, and signals with fanout greater than one must be driven by explicit clocked splitter gates. Furthermore, ensuring correct timing across imbalanced paths is particularly challenging in AQFP, necessitating strict path balancing whereby all inputs to a gate traverse the same number of clocked cells. As a result, AQFP design flows include a post-synthesis technology legalization step that inserts buffers and splitters to enforce path balancing and satisfy fanout constraints~\cite{BS_TCAD_Lee}. This step introduces significant overhead and often dominates circuit area and delay.

Prior work has extensively studied buffer and splitter insertion with the objective of minimizing total JJ count and circuit depth~\cite{BS_TCAD_Fu,BS_TCAD_Lee,bs-2,heuristicASP,DAC22,SOTA,BS_Optimal_AQFP,BS_AQFP_2026,AvilesNPhase}. However, these metrics do not accurately reflect the final physical design cost in AQFP.  In contrast to CMOS, where transistor count is often a reasonable proxy for area due to flexible placement, AQFP layouts are constrained to place each logic level in its own physical row~\cite{AQFPPlacement, Rowbased_Layouts}. Consequently, circuit area is determined not only by the number of inserted elements but also by their spatial distribution across logic levels. As a result, minimizing JJ count alone can lead to suboptimal physical designs by poorly distributing logic assignments as shown in Fig.~\ref{fig:c3540_dist}. Instead, we observe that AQFP circuit area is more accurately captured by the product of circuit width (maximum row width) and circuit depth (row count), which reflects the dominant layout dimensions.

Based on this observation, we redefine the AQFP legalization optimization problem as minimizing the circuit width--depth product during buffer and splitter insertion, yielding a layout-aware optimization objective that more accurately captures physical design area.

More precisely, this article makes the following contributions:
\begin{itemize}
    \item We propose the minimization of the circuit width--depth product during buffer and splitter insertion as a key optimization problem and prove that it is NP-complete.
    \item We develop heuristic algorithms to optimize this layout-aware objective during AQFP technology legalization.
    \item We demonstrate scalable runtime performance and achieve an average 30\% reduction in post-placement area compared to state-of-the-art methods, with on average only a 3\% increase in junction count, and up to 61\% area reduction on individual circuits, validating the proposed objective as a more accurate proxy for physical design cost.
 
\end{itemize}

\section{Background}
 
AQFP is a clocked superconducting logic family that achieves ultra-low-power operation by adiabatically transitioning between an off state and the logic states `1' and `0'~\cite{AQFP}. This transition is driven by an AC excitation current, with a small DC bias, that simultaneously serves as both the power supply and clock. Logical values are encoded by the direction of input currents, which must satisfy both setup and hold timing constraints with respect to the rising edge of the AC clock for every gate~\cite{AQFP_Timing}. When these timing constraints are satisfied, the input currents establish persistent magnetic flux in either the left or right superconducting loop, generating the corresponding output current while the AC excitation remains active.  The fundamental AQFP buffer is illustrated in Fig.~\ref{fig:AQFPBuffer}.

Basic AQFP logic is constructed from Majority-3 (MAJ3) gates, whose output is determined by the majority polarity of three input currents~\cite{scl}. Together with signal inversion, which is realized by appropriately designed transformers, the MAJ3 gate forms a functionally complete logic basis~\cite{maj3}. AQFP libraries therefore commonly consist of four fundamental cell types: clocked buffers, NOT gates, constant cells, and branch cells that implement both majority and splitter functionality~\cite{scl}.  These cell types form the target library for AQFP logic synthesis and technology legalization.

\begin{figure}[t]
\includegraphics[width=0.75\columnwidth]{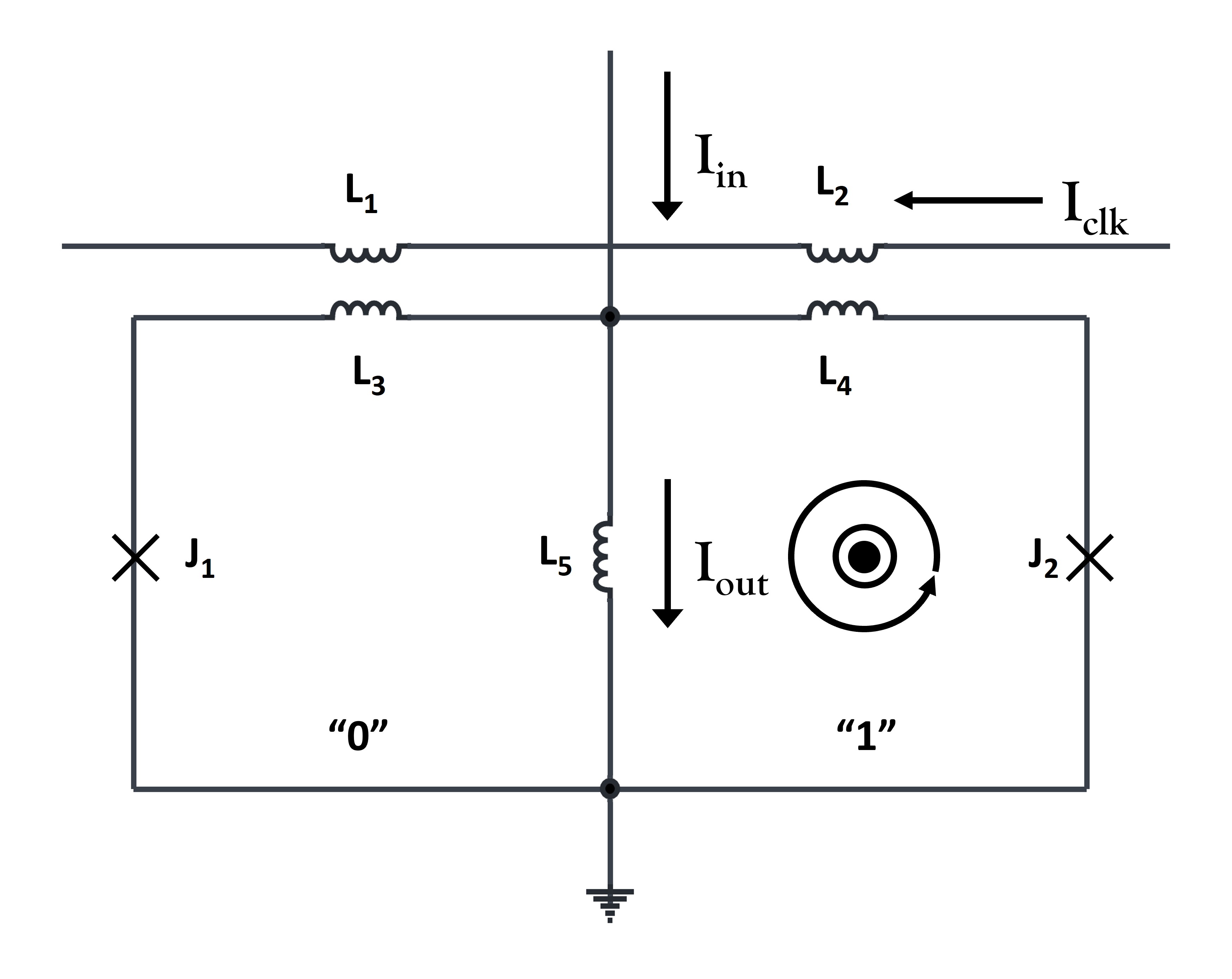}% This is a *.png file
%\vspace{-0.6cm}
\centering
\caption{Circuit Construction of AQFP Buffer.  Current directions shown for when flux is in the right loop, corresponding to a logical `1'.}\label{fig:AQFPBuffer}
\end{figure}

Using this cell library, automated AQFP design flows begin with majority-based logic synthesis using conventional optimization techniques. The synthesized netlist then undergoes {\em technology legalization} by inserting explicit clocked splitter cells to drive multi-fanout signals and clocked buffers to balance reconvergent paths, thereby ensuring correct synchronization across all logic paths. These technology legalization steps introduce substantial overhead and frequently dominate the final cell count. In particular, prior work has focused on minimizing the number of inserted buffers and splitters, or equivalently total Josephson junction (JJ) count, through improved logic-level assignment~\cite{BS_TCAD_Fu,BS_TCAD_Lee,bs-2,heuristicASP,DAC22,SOTA,BS_Optimal_AQFP,BS_AQFP_2026,AvilesNPhase}.  Despite these optimizations, AQFP legalization of an 8-bit multiplier containing only 77 logic gates requires 371 inserted buffers and splitters~\cite{BS_TCAD_Lee}.

After synthesis and technology legalization, each logic level is typically assigned a dedicated placement row, as illustrated in Fig.~\ref{fig:AQFPlayouts}, with all cells in a row sharing a common AC clock line. The AC clock is phase delayed between adjacent rows so that data transfer satisfies setup and hold constraints between source and sink cells~\cite{qPRO}. During placement, additional rows of clocked buffers are frequently introduced to satisfy maximum interconnect length constraints, increasing circuit depth, JJ count and physical area~\cite{AQFPPlacement}.

\begin{figure}[htbp]
\includegraphics[width=0.6\columnwidth]{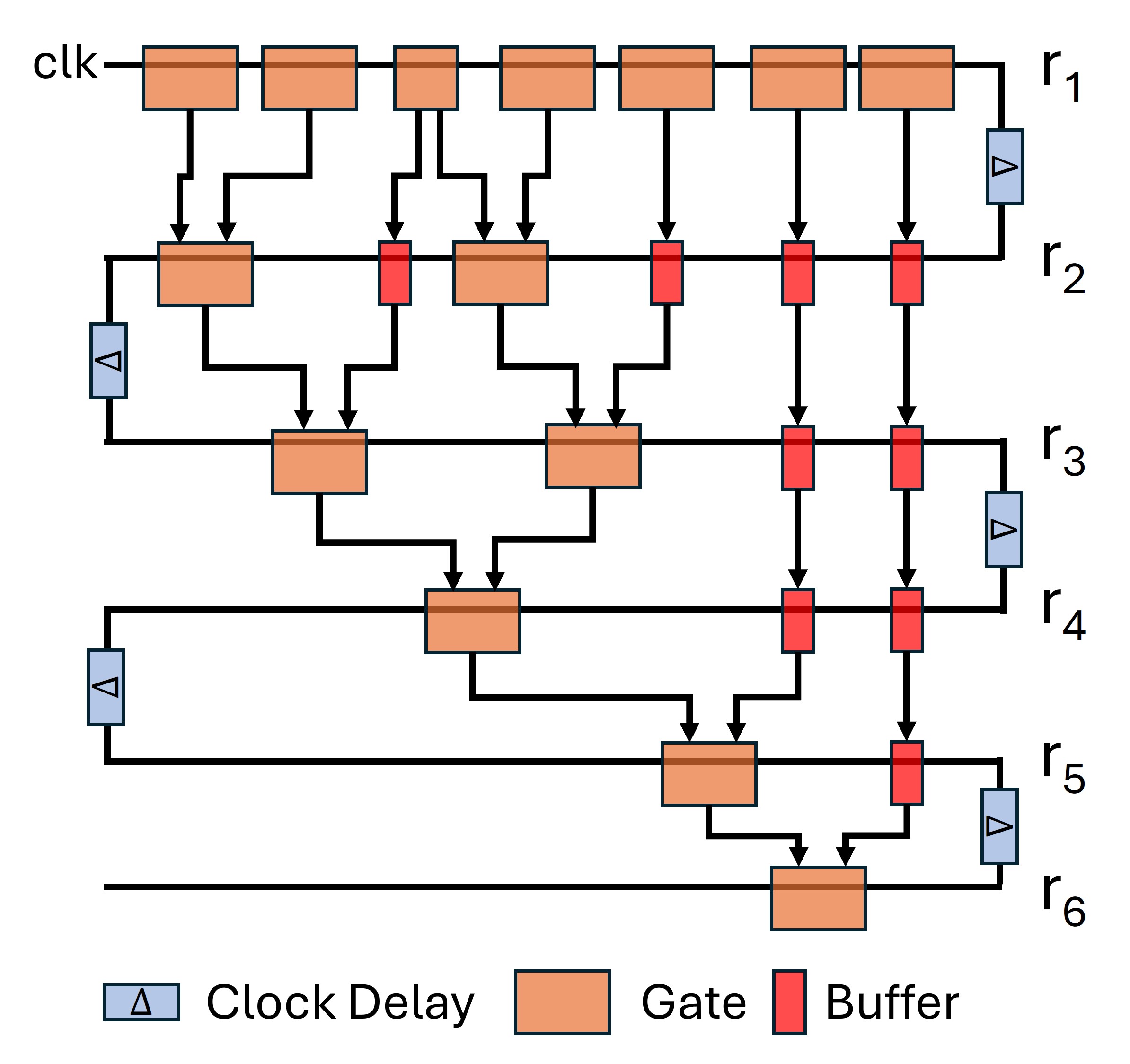}% This is a *.png file
%\vspace{-0.6cm}
\centering
\caption{AQFP placement constrains each logic level to be assigned to a single physical row.}\label{fig:AQFPlayouts}
\end{figure}

\section{Redefining AQFP Technology Optimization}

The row-based organization imposed by AQFP creates a unique coupling between logic-level assignment and physical placement that has remained unexplored.  This observation motivates revisiting AQFP technology legalization as a layout-aware optimization problem, developed in the following subsection.

\subsection{Layout-Aware Optimization}
To directly target physical design cost, we formulate AQFP technology legalization as a layout-aware optimization over the set of row assignments (R), as well as the set of buffers (B) and splitters (S) inserted for legalization.  Given an input logical netlist $\mathcal{G}_L=(\mathrm{L},\mathrm{E})$, the objective is to construct an AQFP netlist $\mathcal{G}=(\mathrm{V},\mathrm{E}')$, where $\mathrm{V}=\mathrm{L}\cup\mathrm{B}\cup\mathrm{S}$, such that all cell library fanout constraints are satisfied, the circuit is fully path balanced, and the product  of maximum row width and circuit depth is minimized.  

Let $C(i)$ denote the physical width of cell $i$, let $r_i$ denote its assigned row, and $s_i$ denote its maximum fanout. Let $\mathcal{D}_{\max}$ denote the maximum allowable circuit depth. The resulting optimization problem is formulated as follows.

\begin{equation}
\label{eq:objective}
\min \quad \mathcal{W}\cdot\mathcal{D}\hspace*{0.15in} \textrm{s.t.}
\end{equation}

\begin{equation}
\label{eq:width}
\mathcal{W}
=
\max_{r}
\left(
\sum_{i\in V:\,r_i=r} C(i)
\right)
\end{equation}

\begin{equation}
\label{eq:Dmax}
\mathcal{D} \leq \mathcal{D}_{\max}
\end{equation}

\begin{equation}
\label{eq:po}
r_i = \mathcal{D},
\qquad \forall i \in \mathrm{O}
\end{equation}

\begin{equation}
\label{eq:pi}
r_i = 0,
\qquad \forall i \in \mathrm{I}
\end{equation}

\begin{equation}
\label{eq:logic_fanout}
|\mathrm{FO}(i)| \le s_i,
\qquad \forall i \in \mathrm{V}
\end{equation}
\begin{equation}
\label{eq:pipeline}
r_j = r_i + 1,
\qquad \forall (i,j)\in E'
\end{equation}

Notably, the minimum feasible circuit depth does not necessarily minimize physical area. Allowing increased circuit depth enlarges the search space for feasible row assignments along critical paths, often enabling reductions in maximum row width and overall area. The proposed formulation therefore explores circuit depths up to a user-specified bound $\mathcal{D}_{\max}$, allowing designers to balance physical implementation quality against latency requirements.

\subsection{Problem Complexity}
The computational complexity of the proposed formulation is established by proving that its corresponding decision problem is NP-complete, from which NP-hardness of the optimization problem follows immediately.

The corresponding decision problem, asks for a given objective bound $K$, if there exists a legalized AQFP netlist satisfying Eqs.~(\ref{eq:width})--(\ref{eq:pipeline}) such that
\[
\mathcal W\mathcal D\le K.
\]

The following theorem establishes NP-completeness of the proposed layout-aware decision problem.

\begin{theorem}
The decision version of the layout-aware AQFP legalization problem is NP-complete.
\end{theorem}

\begin{proof}

To establish NP-hardness we show reduction from the \emph{Minimum Machine Scheduling Problem}, proven NP-complete by Finke \emph{et al.}~\cite{Min_machine_scheduling}. An instance of this problem consists of $(T,\prec,m)$, where $T$ is a set of unit-time tasks, $\prec$ defines precedence constraints, and $m$ is the maximum number of machines active at any time. The decision problem asks whether there exists a minimum-length schedule satisfying $\prec$ that requires at most $m$ machines.

Given an instance $(T,\prec,m)$, we construct an AQFP legalization instance as follows. Fix the allowable circuit depth to the minimum schedule length,
\[
\mathcal D_{\max}=D_{\min},
\]
where \(D_{\min}\) denotes the minimum feasible schedule length of the given scheduling instance.

Then for every task $t_i\in T$, create one logic gate $v_i\in L$ with unit width
\[
C(v_i)=1.
\]
For every precedence relation
\[
t_i \prec t_j,
\]
construct one edge
\[
(v_i,v_j)\in E.
\]

To eliminate legalization effects not present in the scheduling problem, the AQFP cell library is restricted such that every logic gate supports unlimited fanout ($s_i=\infty$) and path-balancing buffers have zero width. With this input cell library, splitter insertion is eliminated and inserted buffers do not contribute to the optimization objective. Under these assumptions, the AQFP legalization problem reduces to minimizing the maximum row width subject only to precedence constraints.

In this mapping, each task becomes a logic gate, each AQFP row corresponds to one scheduling time step, and each row width represents the number of machines required to execute all tasks assigned to that time step. Assigning gate $v_i$ to row $r_i$ is therefore equivalent to scheduling task $t_i$ at time step $r_i$. Since every gate has unit width, the width of row $r$,
\[
\mathcal{W}(r)
=
\sum_{i:r_i=r} C(v_i),
\]
is exactly the number of simultaneously executing tasks, and consequently
\[
\mathcal{W}
=
\max_r \mathcal{W}(r)
\]
is precisely the number of machines required by the schedule.

Therefore, because precedence constraints are preserved exactly by construction, there exists a minimum-depth AQFP legalization having
\[
\mathcal{W}\le m
\]
if and only if the original scheduling instance admits a minimum-length schedule requiring at most $m$ machines. Since the construction is linear in the number of tasks and precedence constraints, this reduction is computable in polynomial time. Hence, the AQFP layout-aware optimization problem is NP-hard.

To establish membership in NP, we first show that the certificate, namely the legalized AQFP netlist, has polynomial size with respect to the input logical netlist.

The legalization process constructs one splitter tree for each multi-fanout signal. A splitter tree driving $k$ outgoing edges contains at most $k-1$ splitters. Summing over all splitter trees therefore yields

\[
|S|
\leq
\sum_i (k_i-1)
<
\sum_i k_i
\le
|E|,
\]

and consequently,

\[
|S| \le |E|.
\]

Likewise, each legalized edge may require a chain of path-balancing buffers whose length is bounded by the maximum possible circuit depth. Because the circuit depth cannot exceed the total number of logic gates and splitters,

\[
\mathcal{D} \le |L|+|S|,
\]

the total number of inserted buffers satisfies

\[
|B|
\le
|E|(|L|+|S|)
\le
|E|(|L|+|E|).
\]

Therefore,

\[
|V|
=
|L|+|S|+|B|
=
O(|E|^2),
\]
and the legalized AQFP netlist is polynomial in the size of the input netlist.

Given a legalized AQFP netlist together with row assignments for all vertices, verification consists of checking that every edge satisfies Eq.~\ref{eq:pipeline}, every vertex satisfies the fanout constraint of Eq.~\ref{eq:logic_fanout}, and computing the maximum row width according to Eq.~\ref{eq:width}. Each of these operations requires only a single traversal of the legalized netlist and therefore executes in linear time with respect to its size. Since the certificate is polynomially bounded with respect to the input, verification likewise requires polynomial time. Hence, the layout-aware AQFP legalization problem belongs to NP.

Combining NP membership with the NP-hardness established above proves that the decision version of the layout-aware AQFP legalization problem is NP-complete.
\end{proof}

Although the proposed formulation reduces to the minimum machine scheduling problem under the simplifying assumptions used in the proof, the full AQFP legalization problem additionally requires simultaneous optimization of path-balancing buffers, splitter tree topology, and circuit depth. These technology-specific constraints fundamentally distinguish layout-aware AQFP legalization from classical scheduling formulations and preclude the direct application of existing scheduling algorithms.

\section{Layout-Aware Optimization Flow}
To address the complexity of the optimization problem, the proposed framework iteratively minimizes circuit width under fixed splitter trees and a fixed circuit depth before exploring alternative splitter topologies and progressively increasing circuit depth, as illustrated in Fig.~\ref{fig:flow}. Each iteration evaluates candidate gate moves, where a gate is reassigned to an adjacent upstream or downstream row such that it minimizes a global width score. 
\begin{figure}[htbp]
\includegraphics[width=0.6\columnwidth]{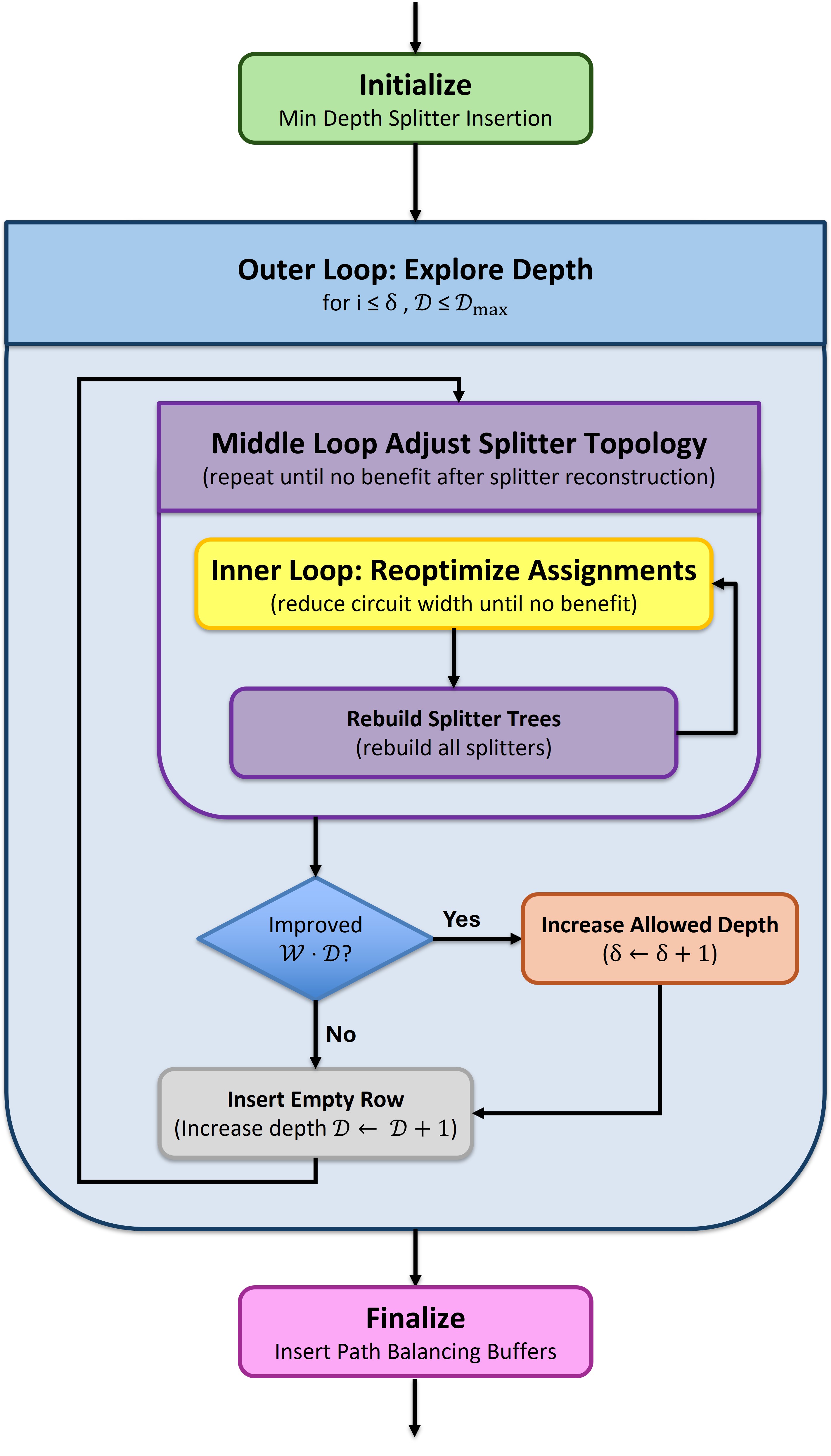}% This is a *.png file
%\vspace{-0.6cm}
\centering
\caption{Layout-Aware top-level flow}\label{fig:flow}
\end{figure}

Gate move selection must account for the impact of all affected Transitive Fanin/Fanout (TFI/TFO) gates as well as any buffers required to maintain legality. Explicitly inserting and removing these buffers after every candidate move would repeatedly modify the graph representation, substantially increasing both graph size and move evaluation cost. Instead, the proposed framework maintains the inserted splitter trees throughout optimization while path-balancing buffer insertion is modeled implicitly until final row assignments have been determined. 

To model the buffer insertion impacts during optimization, we maintain a vector of row widths \textbf{W} of length Depth, where each entry is calculated as: 
\begin{equation}
\label{eq:row_width}
\mathrm{\textbf{W}}[r]
=
\sum_{i \in V:\, r_i = r} C(i)
+
\sum_{\substack{(i,j)\in E \\ r_i < r < r_j}}
\beta
\end{equation}
The first term captures the contribution of cells assigned to row $r$, while the second term accounts for path-balancing buffers, each of width $\beta$, required by edges spanning row $r$.  Our circuit width is then defined as $\mathcal{W}=\max_r \mathbf{W}[r]$ and we define the set of peak rows $P=\left\{r \;\middle|\;\textbf{W}[r]=\mathcal{W}\right\}$.  In our flow each candidate move $i$ results in a modified row width distribution \textbf{$W_i$} and our cost function selects the move with the minimum cost distribution for enactment.

In the remaining section, we detail all of the subproblems before describing their unification into our overall flow, namely: initialization and splitter insertion Sect~\ref{sec:init}, calculation of move impacts Sect~\ref{sec:move_impacts}, ranking candidate moves Sect~\ref{sec:scoring}, and heuristic depth exploration Sect~\ref{sec:depth_increase}.

\begin{figure*}[htbp]
\includegraphics[width=\linewidth]{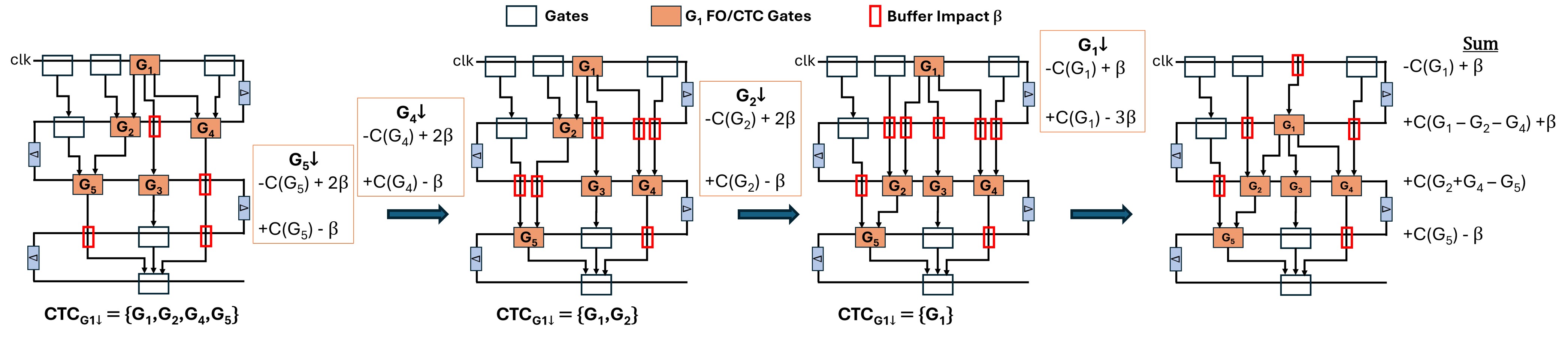}% This is a *.png file
%\vspace{-0.6cm}
\centering
\caption{Decomposition of a downstream gate move ($G_1$) into local impact vectors. Summing the impacts across the critical transitive closure exactly reproduces the global legalization and row-width changes.}\label{fig:Gate_impact}
\end{figure*}

\subsection{Initialization}\label{sec:init}

We begin optimization by constructing a legalized minimum-depth splitter topology and generating an initial row assignment for every gate. The optimization is initialized at the minimum feasible depth, providing the lowest-latency legal implementation while allowing subsequent stages to selectively explore increased circuit depths when they improve the width--depth product.  Fanout constraints are therefore resolved first through splitter insertion, after which multiple minimum-depth row assignments are evaluated to identify the most favorable starting point for width optimization.

In AQFP, minimum-depth splitter tree construction has been extensively studied and can be performed efficiently~\cite{heuristicASP}. We implement the method of~\cite{heuristicASP}, with the modification that path-balancing buffers are not inserted during initialization. Using the resulting splitter topology, we evaluate three alternative row assignment strategies: an as-soon-as-possible schedule ASAP(G), an as-late-as-possible schedule ALAP(G), and a minimum-buffer schedule obtained through linear programming LP(G). Both ALAP and LP schedules are constrained to the minimum depth determined by the ASAP schedule.

The LP formulation is adapted from the minimum-buffer scheduling framework of~\cite{AvilesMultiPhase}. To improve runtime, we constrain the feasible range of each row assignment using ASAP and ALAP bounds. The resulting optimization minimizes the total buffer cost across all edges,

\begin{equation}
\label{eq:lpcost}
\min \sum_{(i,j)\in E} C_{ij}\hspace*{0.15in} \textrm{s.t.}
\end{equation}

subject to

\begin{equation}
\label{eq:cij_cost}
1 \le r_j-r_i \le C_{ij}+1,
\qquad \forall (i,j)\in E,
\end{equation}

\begin{equation}
G_{ASAP}(i)
\le
r_i
\le
G_{ALAP}(i),
\qquad
\forall i \in V,
\end{equation}

Minimizing (\ref{eq:lpcost}) minimizes buffer insertion and consequently JJ count for a fixed splitter topology.
However, in practice, we observe that ASAP and ALAP schedules frequently produce narrower layouts than the minimum-buffer LP solution and lead to better optimization results despite requiring additional buffers. Consequently, rather than selecting the schedule with minimum JJ count, we select the schedule with minimum circuit width to serves as the starting point for the subsequent optimization.

\begin{figure*}[htbp]
\includegraphics[width=\linewidth]{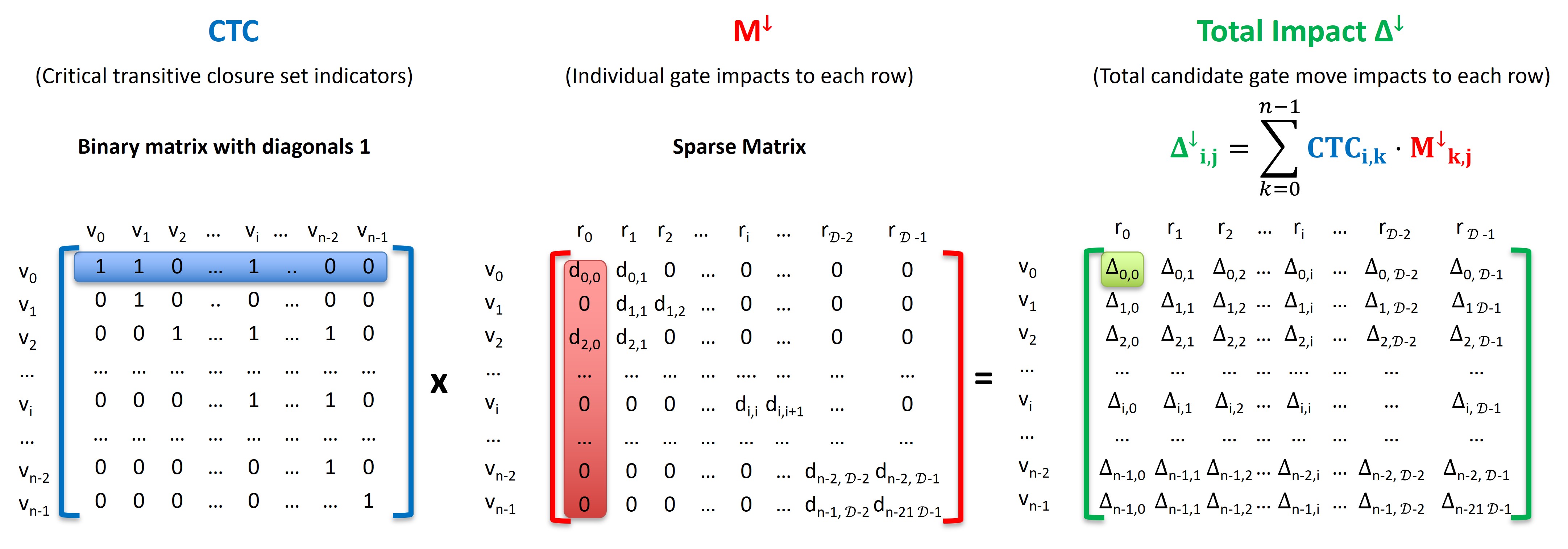}% This is a *.png file
%\vspace{-0.6cm}
\centering
\caption{Matrix formulation for candidate move evaluation. Local gate impacts are accumulated across the critical transitive closure through matrix multiplication to compute the cumulative impact of every downstream (upstream) gate move simultaneously.}\label{fig:matrices}
\end{figure*}

\subsection{Capturing gate move impacts}\label{sec:move_impacts}
During optimization, candidate row assignments are explored through local gate moves, where a gate and any dependent gates are reassigned exactly one row upstream or downstream. To efficiently evaluate the impact of every possible gate move, we develop a framework for calculating move impacts on our row width distribution as a matrix multiplication. 

\subsubsection{Individual gate moves}
The impact of moving a gate $i$ on the row-width vector $\mathbf{W}$ must account for both the reassignment of dependent gates and the resulting changes in buffer occupancy captured by (\ref{eq:row_width}). We refer to the set of gates that must be reassigned during a gate move as the critical transitive closure of $i$. The key observation is that the cumulative impact of a gate move can be decomposed into the sum of the individual impacts of the gates contained in its critical transitive closure.

Moving a gate changes both its own row assignment and the impact of path-balancing buffers on adjacent edges. Consequently, the cost of moving one gate depends on whether neighboring gates also move, introducing dependencies across the critical transitive closure. Concretely, as $i$ is moved downstream a buffer must be added on fanin edges where $i$ moves away from its fanin.  Similarly a buffer must be removed from fanout edges where $i$ moves closer to its fanout $j$.  

To enable exact decomposition of move impacts into local contributions, we encode path-balancing effects directly into the local impact of each gate. Consider a gate $i$ assigned to row $r$, moving $i$ downstream removes width $C(i)$  from row $r$ while increasing the buffer contribution by $|\mathrm{FI}(i)|\beta$ due to newly spanning edges. Simultaneously, row $r+1$ gains width $C(i)$ while reducing the buffer contribution by $|\mathrm{FO}(i)|\beta$. The upstream case is analogous, exchanging fanins and fanouts.

Under this formulation, the local path-balancing buffer costs associated with moving individual gates accumulate to exactly capture the global row-width change induced by a move. Figure~\ref{fig:Gate_impact} illustrates this property for a downstream movement. Gates within the critical transitive closure move together, causing any path-balancing buffers inserted between them to be simultaneously removed and reintroduced. Consequently, the corresponding buffer-width contributions cancel when the local impact vectors are summed. In contrast, edges crossing the boundary of the critical transitive closure experience a net change in path-balancing cost, which is captured entirely by the local impact of the boundary gates, analogous to retiming buffers across the closure.  Unlike classical retiming, however, which relocates registers to optimize timing or register count, the proposed local impact vectors jointly encode the effects of both gate relocation and buffer movement on the layout-aware row-width objective.  As a result, summing the local impact vectors of all gates in the critical transitive closure exactly reproduces the legalization and row-width changes of explicitly moving the entire subgraph, enabling efficient evaluation of all candidate gate moves without graph modification or re-legalization.

We store the impact vectors for all gates in two sparse matrices representing upstream and downstream moves separately, $\mathbf{M}^{\uparrow}$ and $\mathbf{M}^{\downarrow}$, where $\mathbf{M}^{\uparrow}[i,:] = \mathbf{u}_i^{T}$ and $\mathbf{M}^{\downarrow}[i,:] = \mathbf{d}_i^{T}$. Let rows be indexed by $r \in \{0,\ldots,D-1\}$. 
For each gate $i$, define its upstream-impact vector $\mathbf{u}_i \in \mathbb{R}^{D}$ and downstream-impact vector $\mathbf{d}_i \in \mathbb{R}^{D}$ as
\begin{equation}
\label{eq:upstream_impact}
\mathbf{u}_i[r] =
\begin{cases}
+\infty,
&
r_i = 0,
\\[4pt]
-\!C(i) + |\mathrm{FO}(i)|\beta,
&
r = r_i,
\\[4pt]
C(i) - |\mathrm{FI}(i)|\beta,
&
r = r_i - 1,
\\[4pt]
0,
&
\text{otherwise}.
\end{cases}
\end{equation}

\begin{equation}
\label{eq:downstream_impact}
\mathbf{d}_i[r] =
\begin{cases}
+\infty,
&
r_i = \mathcal{D},
\\[4pt]
-\!C(i) + |\mathrm{FI}(i)|\beta,
&
r = r_i,
\\[4pt]
C(i) - |\mathrm{FO}(i)|\beta,
&
r = r_i + 1,
\\[4pt]
0,
&
\text{otherwise}.
\end{cases}
\end{equation}
The $+\infty$ cases enforce depth bounds as any moves that could exceed the bounds will be assigned infinite width and thus be rejected as a solution.
The matrices $\mathbf{M}^{\uparrow/\downarrow}$ encode isolated gate impacts which must be summed across the critical transitive closure.

\subsubsection{Critical closure identification}
The critical fanouts of a gate $i$ are defined as
\begin{equation}
\label{eq:critical_fanout}
\mathrm{CFO}(i)
=
\left\{
j \in \mathrm{FO}(i)
\;\middle|\;
r_j = r_i + 1
\right\}.
\end{equation}

Similarly, the critical fanins are defined as
\begin{equation}
\label{eq:critical_fanins}
\mathrm{CFI}(i)
=
\left\{
j \in \mathrm{FI}(i)
\;\middle|\;
r_i = r_j -  1
\right\}.
\end{equation}
Note that iff $j \in \mathrm{CFO}(i)$ then $i \in \mathrm{CFI}(j)$, thus the critical fanins of all gates can be computed directly from the critical fanouts.

To efficiently evaluate gate moves, we define a critical transitive closure (CTC) matrix. For each gate $i$, a critical closure vector $\mathbf{c}_i$ identifies the set of gates whose row assignments must be modified when $i$ is moved downstream. Formally, $\mathbf{c}_i[j]=1$ if gate $j$ belongs to the critical transitive closure of $i$, and $\mathbf{c}_i[j]=0$ otherwise.
These closure vectors are accumulated in reverse topological order using the critical fanout relation 
\begin{equation}
\label{eq:closure_recursion}
\mathbf{c}_i
=
\mathbf{e}_i
\;\vee\;
\bigvee_{j\in\mathrm{CFO}(i)}
\mathbf{c}_j,
\end{equation} where $\mathbf{e}_i$ is the unit basis vector whose $i$-th entry is one and all others are zero, and $\vee$ denotes element-wise logical OR.  All closure vectors are assembled into the critical transitive closure matrix \textbf{CTC}. 

\begin{equation}
\label{eq:ctc}
\mathbf{CTC}
=
\begin{bmatrix}
\mathbf{c}_0^T
\\
\mathbf{c}_1^T
\\
\vdots
\\
\mathbf{c}_{n-1}^T
\end{bmatrix}.
\end{equation}
Due to the duality between critical fanout and critical fanin dependencies, the transpose $\mathbf{CTC^T}$ directly encodes the corresponding closure relations for upstream moves.

\subsubsection{Matrix operations}
Having stored the local gate impacts in matrices $\mathbf{M}^{\uparrow}$ and $\mathbf{M}^{\downarrow}$ and the corresponding critical transitive closures in $\mathbf{CTC}$, the cumulative impact of every candidate gate move can be computed simultaneously using matrix multiplication. Figure~\ref{fig:matrices} illustrates this transformation, where each row of $\mathbf{CTC}$ identifies the gates participating in a candidate move, and the matrix product accumulates their local impact vectors to produce the corresponding row of $\mathbf{\Delta}^{\uparrow}$ or $\mathbf{\Delta}^{\downarrow}$.

\begin{equation}
\label{eq:delta}
\mathbf{\Delta}^{\uparrow}
=
\mathbf{CTC}^{T}\cdot\mathbf{M}^{\uparrow},
\qquad
\mathbf{\Delta}^{\downarrow}
=
\mathbf{CTC}\cdot\mathbf{M}^{\downarrow}
\end{equation}

Consequently, the cumulative impact matrices $\mathbf{\Delta}^{\uparrow}$ and $\mathbf{\Delta}^{\downarrow}$ are obtained through two global matrix multiplications that simultaneously evaluate the legalization and row-width impacts of every upstream and downstream candidate move.

The local impact matrices $\mathbf{M}^{\uparrow}$ and $\mathbf{M}^{\downarrow}$ each have dimensions $|V|\times\mathcal{D}$ and are constructed in linear time. Since every row contains at most two non-zero entries, they are naturally represented as sparse matrices requiring only $O(|V|)$ storage.

The critical transitive closure matrix $\mathbf{CTC}$, has dimensions $|V|\times|V|$, constructed in $O(|V||E|)$ time by propagating closure vectors across the graph.  However, representing it as a bitmap enables highly efficient vectorized logical OR operations. 

The proposed framework formulates a local impact model targeting row width whose effects compose exactly over critical closures, allowing global impacts of every candidate move to be evaluated in parallel through sparse matrix operations instead of repeated graph legalization. 

Each row of the resulting $\mathbf{\Delta}^{\uparrow}$ and $\mathbf{\Delta}^{\downarrow}$ matrices represents the cumulative change in the row-width vector $\mathbf{W}$ produced by moving a candidate gate upstream or downstream, respectively.  Consequently, evaluating a candidate move requires only adding the corresponding impact vector to the current row-width vector,

\begin{equation}\label{eq:Wi_calc}
\mathbf{W}_i^{\uparrow}
=
\mathbf{W}
+
\mathbf{\Delta}^{\uparrow}[i,:],
\qquad
\mathbf{W}_i^{\downarrow}
=
\mathbf{W}
+
\mathbf{\Delta}^{\downarrow}[i,:].
\end{equation}

The candidate move with the minimum cost is then selected using the scoring mechanism described in the following section which prioritizes reduction in row width as well as density around peak rows.

\subsection{Scoring and Selecting Candidate Moves}\label{sec:scoring}

Once the row-width impacts of all candidate moves have been computed, the remaining challenge is selecting the move most beneficial to final circuit area. While minimizing the maximum row width is the primary objective, candidate moves often produce identical maximum widths or merely redistribute congestion between neighboring rows without improving the overall layout. To distinguish among these plateau solutions, we introduce a critical row-density metric $\rho$ that penalizes congestion surrounding each peak row. This approach is analogous to timing-driven placement, where optimization based solely on the single worst critical path is often insufficient, motivating objectives that account for the sensitivity of near-critical nets during optimization~\cite{Sensitivity_Nets,DREAMPlace4}. 

We define the critical threshold $\tau$ as a value located a fraction $\sigma$ of the way between the ideal width and the current width, where the ideal width $\mathcal{W}^*$ is fixed upon initialization as the average row width obtained under perfectly uniform row utilization. 

\begin{equation}\label{eq:tau}
\tau = \mathcal{W}^* + \sigma(\mathcal{W}- \mathcal{W}^*) 
\end{equation}
\begin{equation}\label{eq:ideal_w}
\mathcal{W}^* = \frac{1}{\mathcal{D}}\sum_r \mathbf{W}[r]
\end{equation}

To penalize critical density we define a weighting vector $\mathbf{\omega_p}$ on \textbf{W} for each peak where $\omega_{max}$ is the maximum row weight and $\delta\omega$ is the linear decay applied per row away from the peak
\begin{equation}
\label{eq:omega_p}
\mathbf{\omega_p}[r] =
\begin{cases}
\max\!\bigl(1,\,
\omega_{\max}
-
\delta\omega\,|p-r|
\bigr),
&
\mathrm{LB}_p \le r \le \mathrm{UB}_p,
\\[4pt]
0,
&
\text{otherwise}.
\end{cases}
\end{equation}
\begin{equation}
\label{eq:lb_p}
\mathrm{LB}_p
=
\min \left\{
r \,\middle|\,
\forall k \in [r,p],\ \mathrm{\textbf{W}}[k] > \tau
\right\}
\end{equation}
\begin{equation}
\label{eq:ub_p}
\mathrm{UB}_p
=
\max \left\{
r \,\middle|\,
\forall k \in [p,r],\ \mathrm{\textbf{W}}[k] > \tau
\right\}
\end{equation}
Using this weighting vector we compute our critical density as 
\begin{equation}\label{eq:crit_density}
\rho=\sum_{p\in P}\boldsymbol{\omega_p^T}
\left(\mathbf{W}-\tau\mathbf{1}\right)^{2}
\end{equation} where the square operation is performed element-wise.

Figure~\ref{fig:rho} illustrates the computation of the critical density metric for an example row-width distribution with $\omega_{\max}=20$ and $\delta\omega=2$. Only rows whose widths exceed the critical threshold $\tau$ contribute to the metric, defining the congestion cluster surrounding each peak row. Within this region, rows nearer the peak are assigned larger weights, encouraging optimization to reduce congestion locally rather than simply shifting it between neighboring rows. At the same time, the quadratic penalty ensures that significantly over-utilized rows continue to dominate the objective even if they are farther from the peak. Consequently, the proposed metric balances proximity to peak rows with the severity of local congestion.

\begin{figure*}[htbp]
\includegraphics[width=\linewidth]{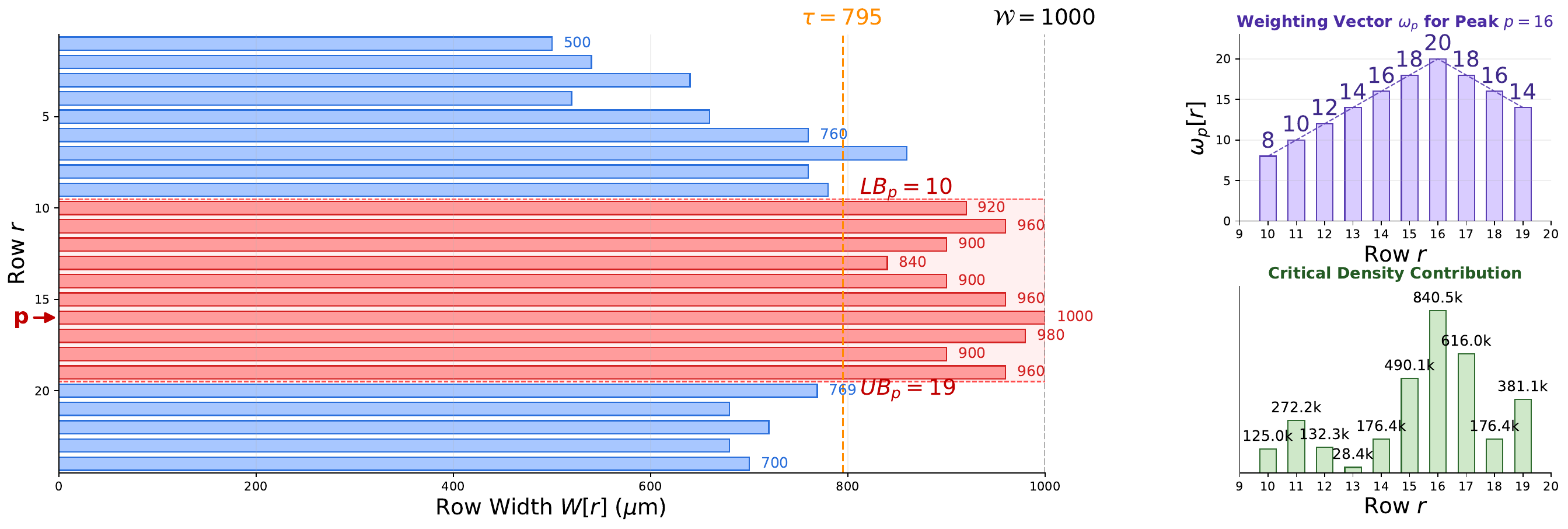}% This is a *.png file
%\vspace{-0.6cm}
\centering
\caption{Illustration of the proposed critical density metric. Row contributions are computed relative to the critical threshold $\tau$ using linearly decaying weights centered around the peak row p=16.}\label{fig:rho}
\end{figure*}

The critical density metric is incorporated into the candidate move evaluation through a lexicographic scoring function. Each candidate gate move is evaluated by computing the resulting cost tuples
\[
\langle\mathcal{W}_i^{\uparrow}, |\mathrm{P}_i^{\uparrow}|, \rho_i^{\uparrow}\rangle
\quad\text{and}\quad
\langle\mathcal{W}_i^{\downarrow}, |\mathrm{P}_i^{\downarrow}|, \rho_i^{\downarrow}\rangle
\]

The lexicographically minimum tuple determines the selected gate move and direction, prioritizing reductions in maximum row width, followed by the number of peak rows $|P|$, and finally the critical density $\rho$. Once the resulting row-width vector $\mathbf{W}_i$ has been computed for each candidate move, all move scores can be evaluated in $O(|V||P|\mathcal{D})$ time, since evaluating $\rho$ requires summing across the $\mathcal{D}$-dimensional row-width vector for each peak row. In practice, $|P|$ is typically small, although in the worst case every row may be a peak row, yielding a complexity of $O(|V|\mathcal{D}^2)$.

Combining move-impact computation with candidate scoring results in an overall worst-case complexity of $O(|V||E|+|V|\mathcal{D}^2)$. However, because $|P|$ is typically small and $\mathcal{D}\ll|E|$, the runtime is dominated in practice by the $O(|V||E|)$ move-impact computation. The proposed scoring function enables efficient evaluation of all candidate moves while preserving width reduction as the primary optimization objective, using critical density to prioritize among candidate moves with comparable width by favoring those that disperse local congestion and create greater opportunity for subsequent width reduction.

\subsection{Depth Exploration}\label{sec:depth_increase}

Introducing an additional row increases row-assignment flexibility by allowing gates on previously critical paths to become movable. Accordingly, after width optimization converges at a given depth, we explore solutions with progressively increasing depth.

A simple approach could insert an empty row at an arbitrary location. However, because width reduction is primarily driven by gate movement around peak rows, we use a lightweight heuristic to place the additional row near regions where congestion is most likely to be reduced. Specifically, we estimate whether the dominant peak row is more likely to benefit from upstream or downstream gate movements.

Moving a gate $v$ upstream introduces $|\mathrm{FO}(v)|$ path-balancing buffers, while moving it downstream introduces $|\mathrm{FI}(v)|$ buffers. Consequently, we define a directional bias for peak row $p$ as

\begin{equation}\label{eq:down}
\mathrm{Dir\_bias}=\sum_{v \in V_{r_{\mathrm{p}}}}\left(\left|\mathrm{FO}(v)\right|-\left|\mathrm{FI}(v)\right|\right)
\end{equation}

Intuitively, splitter-dominated rows tend to favor downstream movement, whereas logic-dominated rows tend to favor upstream movement. The directional bias therefore provides a simple estimate of which side of the peak is likely to offer greater row-assignment flexibility.

Using this estimate, the empty row is inserted immediately outside the critical region surrounding the dominant peak:
\begin{equation}\label{eq:down_end}
r\_{\mathrm{empty}} =
\begin{cases}
\min(\mathrm{UB_p} + 1, \mathcal{D}-1), & \text{if } \mathrm{Dir\_bias} > 0, \\
\max(\mathrm{LB_p} - 1,0), & \text{otherwise.}
\end{cases}
\end{equation}

All gates assigned to row $r\_{\mathrm{empty}}$, together with their critical transitive fanouts, are then moved one row downstream, leaving an empty row at $r\_{\mathrm{empty}}$. The resulting increase in depth provides additional flexibility for subsequent width-optimization iterations while preserving legality.

\subsection{Cohesive Flow}

The proposed framework combines the preceding formulations into a unified optimization algorithm that jointly explores row assignments, splitter topology, and circuit depth. Rather than optimizing these components independently, the framework alternates between localized improvements and global depth exploration, progressively refining the solution while maintaining a legal AQFP design.  Algorithm~\ref{alg:layout_aware_optimization} summarizes this flow, which is organized as three nested optimization loops.

Lines 1--5 perform the initialization described in Section~\ref{sec:init}, constructing a minimum-depth legalized splitter topology and selecting the minimum-width row assignment from the ASAP, ALAP, and LP initializations.

Lines 6--39 implement the depth exploration procedure of Section~\ref{sec:depth_increase}. At each explored depth, the current row-width vector and optimization score are initialized (Lines 8--10). The middle loop alternates between splitter-tree reconstruction and gate-move optimization, while the inner loop repeatedly applies the minimum-cost gate move at the current splitter topology.

Within the fixed-depth optimization, Lines 17--23 compute the matrix-based move impacts of Section~\ref{sec:move_impacts}, while Lines 24--30 rank all candidate moves using the lexicographic scoring procedure of Section~\ref{sec:scoring} and apply the minimum-cost move. When no improving move remains, the splitter topology is rebuilt (Line 32) and fixed-depth optimization resumes under the updated topology. The fixed-depth optimization terminates only when neither additional gate moves nor splitter-tree reconstruction produces further improvement.

After convergence at the current depth, the best solution is retained if its width--depth product improves upon the current optimum (Lines 34--37), after which an empty row is inserted to increase the explored circuit depth (Line 39).  Depth exploration terminates after the maximum allowable depth is reached or after $\delta$ non-improving depth expansions, at which point path-balancing buffers are inserted to complete AQFP legalization before passing the design to downstream physical design tools. The parameter $\delta$ prevents unnecessary depth exploration once no further improvement is observed.

Throughout optimization, every accepted move strictly reduces the global optimization objective while maintaining a legal AQFP row assignment. The search therefore remains monotonic, preserves a feasible solution at every iteration, and terminates after a finite number of improving moves at each explored depth. Together with the global matrix-based evaluation of candidate moves, these properties enable efficient exploration of large optimization spaces, allowing the proposed framework to scale to circuits substantially larger than current AQFP fabrication capabilities, as demonstrated in Section~\ref{sec:scalablity}.

\begin{algorithm}[htbp]
\caption{LAYOUT\_AWARE\_OPTIMIZATION$(\mathrm{G},\mathrm{S},\mathrm{W},\delta)$}
\label{alg:layout_aware_optimization}
\begin{algorithmic}[1]

\Statex \textbf{Input:}
\Statex $\mathrm{G}=(\mathrm{V},\mathrm{E})$: logical graph representation of the circuit
\Statex $\mathrm{S}$: maximum splitter fanout
\Statex $\mathrm{C}$: Cell library mapping from vertex $\mathrm{v} \in \mathrm{V}$ to cell width
\Statex $\mathcal{D}_{\max}$: maximum circuit depth
\Statex $\delta$:  Depth exploration without benefit before termination

\Statex \textbf{Output:}
\Statex $\mathrm{G}'$: AQFP legalized design with minimized circuit width--depth product

\State \textit{Initialize: minimum-depth legalization}
\State $\mathrm{G} \gets \mathrm{Insert\_Splitters}(\mathrm{G},\mathrm{S})$ \Comment{legalize fanout}

\State $\mathrm{G} \gets \mathrm{MinArea}(\mathrm{ASAP}(\mathrm{G}), \mathrm{ALAP}(\mathrm{G}), \mathrm{LP}(\mathrm{G}))$
\State $\mathrm{G}_{\mathrm{best}} \gets \mathrm{G}$
\State $i \gets 0$
\State \textit{Outer loop: explore depth}
\While{$i \leq \delta$, $\mathcal{D}\leq\mathcal{D}_{\max}$}
    \State Compute \textbf{W}\Comment{Eq.~\ref{eq:row_width}}
    \State Fix ideal width $\mathcal{W^*}$\Comment{Eq.~\ref{eq:ideal_w}}
    \State $\mathrm{Best\_Score} \gets \langle \mathcal{W}, |P|, \rho \rangle$ \Comment{Eq.~\ref{eq:crit_density}}

    \State $\mathrm{rebuild\_splitters} \gets 1$
    \State \textit{Middle loop: adjust splitter topology}
    \While{$\mathrm{rebuild\_splitters}$}
        
        \State $\mathrm{rebuild\_splitters} \gets 0$
        \State $\mathrm{improved} \gets 1$
        \State \textit{Inner loop: reduce circuit width}
        \While{$\mathrm{improved}$}
            \State $\mathrm{improved} \gets 0$
            \State $\mathrm{Populate}(\mathbf{CTC}, \mathbf{M}^{\uparrow}, \mathbf{M}^{\downarrow})$\Comment{Eqs.~\ref{eq:upstream_impact}--\ref{eq:ctc}}
            \State $\mathbf{\Delta}^{\uparrow} \gets \mathbf{CTC}^{T}\times\mathbf{M}^{\uparrow}$
            \State $\mathbf{\Delta}^{\downarrow} \gets \mathbf{CTC}\times\mathbf{M}^{\downarrow}$
            \State $\mathbf{W}_i^{\uparrow}=\mathbf{W}+\mathbf{\Delta}^{\uparrow}[i,:]$\Comment{Eq.~\ref{eq:Wi_calc}}
            \State $\mathbf{W}_i^{\downarrow}=\mathbf{W}+\mathbf{\Delta}^{\downarrow}[i,:]$\Comment{Eq.~\ref{eq:Wi_calc}}
            \State $(\mathrm{i}^{\star}, \mathrm{dir}^{\star}, \mathrm{Score}^{\star}) \gets \mathrm{Best\_Move}(\mathbf{W}^\uparrow, \mathbf{W}^\downarrow)$

            \If{$\mathrm{Score}^{\star} \prec \mathrm{Best\_Score}$}
                \State $\mathrm{Best\_Score} \gets \mathrm{Score}^{\star}$
                \State $\mathrm{G} \gets \mathrm{Apply\_Move}(\mathrm{G}, \mathrm{i}^{\star}, \mathrm{dir}^{\star})$
                \State $\mathrm{improved} \gets 1$
                \State $\mathrm{rebuild\_splitters} \gets 1$
            \EndIf
        \EndWhile
        \State $\mathrm{G} \gets \mathrm{Rebuild\_All\_Splitters}(\mathrm{G})$
    \EndWhile
    \If{$\mathcal{D}\cdot \mathcal{W} < \mathcal{D}_{\mathrm{best}}\cdot \mathcal{W}_{\mathrm{best}}$}
        \State $\mathrm{G}_{\mathrm{best}} \gets \mathrm{G}$
        \State $\delta \gets \delta +1$
    \EndIf

    \State $i \gets i + 1$
    \State $\mathcal{D}+1$, Empty row insertion \Comment{Eq.~\ref{eq:down_end}}
\EndWhile

\State $\mathrm{G'} \gets \mathrm{Insert\_Path\_Balancing\_Buffers}(\mathrm{G}_{\mathrm{best}})$
\State \Return $\mathrm{G'}$

\end{algorithmic}
\end{algorithm}

\section{Results}

We compare against the state-of-the-art AQFP buffer and splitter insertion framework from TCAD'25~\cite{BS_TCAD_Lee}. Notably, two legalization frameworks were published concurrently in TCAD'25~\cite{BS_TCAD_Lee,BS_TCAD_Fu}, achieving within 0.2\% average JJ count of one another across standard benchmarks. We select~\cite{BS_TCAD_Lee} as our primary baseline because it publicly released its legalized netlists, enabling direct and reproducible comparisons.

To achieve direct comparisons, our optimization begins from the legalized netlists of~\cite{BS_TCAD_Lee} after removing all inserted buffers and splitters, thereby preserving identical logical starting points. Both the baseline netlists from~\cite{BS_TCAD_Lee} and the optimized netlists produced by our flow are then passed through the same physical design stage using an identical AQFP cell library (Table~\ref{tab:celllib}) and placement methodology~\cite{AQFPPlacement}, ensuring that all reported differences arise solely from technology legalization.  To facilitate reproducibility and future comparisons, all legalized AQFP benchmark netlists generated in this work are publicly available~\cite{LayoutAwareAQFPGitHub}.  In our experiments move scoring hyper parameters are set as $\omega_{\max}=20, \delta\omega = 2, \sigma=60\%$.  Runtimes for the proposed framework are reported on an AMD EPYC 7763 Milan processor.  

\begin{table}[htbp]
\centering
\caption{AQFP cell library used for all designs.}
\label{tab:celllib}
\begin{tabular}{lccc}
\toprule
\textbf{Cell Type} &
\textbf{JJs} &
\textbf{Width ($\mu$m)} &
\textbf{Max. Interconnect ($\mu$m)}\\
\midrule
Logic Gate & 6 & $60\times60$ & 300\\
Buffer     & 2 & $20\times40$ & 800\\
Splitter-2 & 2 & $40\times60$ & 200\\
Splitter-3 & 2 & $60\times60$ & 200\\
Splitter-4 & 2 & $80\times60$ & 200\\
\bottomrule
\end{tabular}
\end{table}

\begin{table*}[t]
\centering
\caption{Comparison of JJ-minimized TCAD'25 designs and layout-aware optimization across benchmark circuits, including runtime.}
\label{tab:layout_depth_compare_runtime}
\scriptsize
\setlength{\tabcolsep}{3pt}
\renewcommand{\arraystretch}{1.08}
\begin{adjustbox}{max width=\textwidth}
\begin{tabular}{l c c c c c c c c c c c c c c c c}
\toprule
\multirow{2}{*}{\textbf{Benchmarks}} &
\multirow{2}{*}{\textbf{Gates}} &
\multirow{2}{*}{\makecell{\textbf{Min.}\\\textbf{Depth}}} &
\multicolumn{3}{c}{\textbf{JJ Min--TCAD'25~\cite{BS_TCAD_Lee}}} &
\multicolumn{5}{c}{\textbf{Layout-Aware Depth Min}} &
\multicolumn{6}{c}{\textbf{Layout-Aware Depth Optimized}} \\
\cmidrule(lr){4-6} \cmidrule(lr){7-11} \cmidrule(lr){12-17}
& & &
\textbf{JJs} & \textbf{Width} & \textbf{Runtime(s)} &
\textbf{JJs} & \textbf{Width} & \makecell{\textbf{$\mathcal{WD}$}\\\textbf{Ratio*}} & \textbf{JJ Ratio} & \textbf{Runtime(s)} &
\textbf{JJs} & \textbf{Depth} & \textbf{Width} & \makecell{\textbf{$\mathcal{WD}$}\\\textbf{Ratio*}} & \textbf{JJ Ratio} & \textbf{Runtime(s)} \\
\midrule
\texttt{adder1}    & 7    & 8   & 74    & 140   & 0.00& 78    & 120   & 0.86 & 1.05 & 0.00& 78    & 8   & 120   & 0.86 & 1.05 & 0.00\\
\texttt{adder8}    & 77   & 33  & 1,204  & 980   & 0.01& 1,556  & 560   & 0.57 & 1.29 & 0.01& 1,562  & 35  & 520   & 0.56 & 1.30 & 0.01\\
\texttt{mult8}     & 439  & 70  & 6,014  & 2,380  & 0.18& 7,962  & 1,600  & 0.67 & 1.32 & 0.06& 7,828  & 74  & 1,440  & 0.64 & 1.30 & 0.16\\
\texttt{counter16} & 29   & 17  & 304   & 480   & 0.00& 358   & 440   & 0.92 & 1.18 & 0.00& 378   & 18  & 360   & 0.79 & 1.24 & 0.00\\
\texttt{counter32} & 82   & 23  & 800   & 960   & 0.01& 912   & 880   & 0.92 & 1.14 & 0.01& 942   & 24  & 720   & 0.78 & 1.18 & 0.01\\
\texttt{counter64} & 195  & 30  & 1,864  & 1,920  & 0.02& 2,114  & 1,760  & 0.92 & 1.13 & 0.01& 2,172  & 31  & 1,440  & 0.78 & 1.17 & 0.03\\
\texttt{counter128}& 428  & 38  & 4,062  & 3,840  & 0.07& 4,680  & 3,520  & 0.92 & 1.15 & 0.03& 4,778  & 39  & 2,880  & 0.77 & 1.18 & 0.11\\
\texttt{c17}       & 6    & 5   & 60    & 200   & 0   & 62    & 160   & 0.80 & 1.03 & 0.00& 62& 5   & 160   & 0.80 & 1.03 & 0.00\\
\texttt{c432}      & 121  & 37  & 2,404  & 1,080  & 0.02& 2,452  & 1,060  & 0.98 & 1.02 & 0.01& 2,564  & 39  & 980   & 0.96 & 1.07 & 0.01\\
\texttt{c499}      & 387  & 29  & 4,668  & 3,740  & 0.09& 4,704  & 3,300  & 0.88 & 1.01 & 0.03& 5,088  & 32  & 2,520  & 0.74 & 1.09 & 0.09\\
\texttt{c880}      & 306  & 40  & 4,858  & 2,460  & 0.15& 5,528  & 1,680  & 0.68 & 1.14 & 0.06& 5,528  & 40  & 1,680  & 0.68 & 1.14 & 0.07\\
\texttt{c1355}     & 389  & 29  & 4,702  & 4,120  & 0.06& 4,782  & 3,340  & 0.81 & 1.02 & 0.04& 5,544  & 36  & 2,120  & 0.64 & 1.18 & 0.11\\
\texttt{c1908}     & 289  & 34  & 4,202  & 2,340  & 0.09& 4,298  & 1,920  & 0.82 & 1.02 & 0.02& 4,584  & 37  & 1,640  & 0.76 & 1.09 & 0.07\\
\texttt{c2670}     & 368  & 28  & 6,032  & 3,920  & 0.32& 7,664  & 3,500  & 0.89 & 1.27 & 0.04& 7,766  & 29  & 3,200  & 0.85 & 1.29 & 0.04\\
\texttt{c3540}     & 794  & 52  & 8,650  & 7,400  & 0.81& 11,982 & 3,320  & 0.45 & 1.39 & 0.23& 12,290 & 54  & 3,160  & 0.44 & 1.42 & 0.45\\
\texttt{c5315}     & 1,302 & 40  & 19,092 & 13,840 & 2.06& 21,442 & 6580  & 0.48 & 1.12 & 0.66& 21,664 & 41  & 6,340  & 0.47 & 1.13 & 0.96\\
\texttt{c6288}     & 1,870 & 179 & 28,514 & 3,340  & 2.56& 29,120 & 2,640  & 0.79 & 1.02 & 0.28& 29,558 & 181 & 2,400  & 0.73 & 1.04 & 1.21\\
\texttt{c7552}     & 1,394 & 56  & 23,238 & 9,540  & 4.20& 28,904 & 6,740  & 0.71 & 1.24 & 0.93& 28,904& 56& 6,740& 0.71 & 1.24& 1.26\\
\texttt{sorter32}  & 480  & 30  & 3,840  & 1,920  & 0.06& 3,840  & 1,920  & 1.00 & 1.00 & 0.03& 4,028  & 32  & 1,600  & 0.89 & 1.05 & 0.11\\
\texttt{sorter48}  & 880  & 35  & 7,040  & 3,520  & 0.20& 7,200  & 3,520  & 1.00 & 1.02 & 0.05& 7,910  & 40  & 2,640  & 0.86 & 1.12 & 0.34\\
\texttt{alu32}     & 1,513 & 169 & 36,750 & 9,340  & 2.74& 42,660 & 4,460  & 0.48 & 1.16 & 3.39& 43,224 & 171 & 4,100  & 0.44 & 1.18 & 6.11\\
\midrule
\multicolumn{3}{l}{\textbf{Average Results \& Total Runtime}} &
\multicolumn{2}{c}{} &
\multicolumn{1}{c}{\textbf{13.65}} &
\multicolumn{2}{c}{} &
\textbf{0.79} & \textbf{1.13} &
\multicolumn{1}{c}{\textbf{5.89}} &
\multicolumn{3}{c}{} &
\textbf{0.72} & \textbf{1.17} &
\multicolumn{1}{c}{\textbf{11.15}} \\
\bottomrule
\end{tabular}
\end{adjustbox}
\caption*{\footnotesize\raggedright *Ratio values are computed as Layout-Aware / TCAD'25. Ratios below 1 indicate lower width--depth products and improved layout quality.}
\end{table*}

\subsection{Post-Synthesis Comparison}

Table~\ref{tab:layout_depth_compare_runtime} compares the proposed layout-aware optimization against the state-of-the-art AQFP legalization framework from TCAD'25 ~\cite{BS_TCAD_Lee}. Under a minimum-depth constraint, the proposed optimization reduces the circuit width--depth product by an average of 21\%, exceeding a $2\times$ improvement for \texttt{c3540}, \texttt{c5315}, and \texttt{alu32}. Enabling depth exploration further improves the average reduction to 28\%, demonstrating that modest increases in circuit depth provide additional flexibility for reducing row congestion.  Because the proposed objective differs fundamentally from minimum-JJ legalization, the initialization stage selects the minimum-area legalization from ASAP, ALAP, and minimum-buffer LP schedules. Across the benchmark suite, the LP initialization was selected for only 8 of the 21 circuits, validating the inclusion of additional initialization approaches. 

Despite optimizing a complex layout-aware objective, the proposed framework achieves runtimes comparable to, and frequently lower than, the legalization algorithm of TCAD'25. These results demonstrate that the substantial improvements from the proposed layout-aware objective can be achieved without sacrificing computational efficiency.

\subsection{Post-Placement Comparison}

To validate the proposed optimization objective, we compare the resulting legalized netlists after physical design using the placement method 
from~\cite{AQFPPlacement} for all benchmark circuits in Table~\ref{tab:placement_results}. The average post-synthesis reduction in width--depth product of 28\% closely predicts the measured average post-placement area reduction of 30\%, validating the proposed width--depth product as an accurate post-synthesis proxy for physical design cost.  

Notably, the placement stage frequently inserts additional rows of buffers to satisfy maximum interconnect length constraints, increasing both circuit depth and JJ count. By reducing congestion around highly utilized rows (Fig.~\ref{fig:c3540_dist}), the proposed layout-aware optimization significantly alleviates these physical design challenges. As a result, layout-aware circuits achieve an average 2\% reduction in circuit depth, and 9 benchmark circuits require fewer total JJs than the TCAD'25 baseline~\cite{BS_TCAD_Lee}, despite our optimization targeting area rather than junction count or depth. While the post-synthesis results indicate an average JJ overhead of 17\%, this overhead is reduced to only 3\% after placement due to the improved row distribution requiring fewer additional buffers than~\cite{BS_TCAD_Lee} during placement.  Accordingly, our flow achieves on average a 30\% reduction in area for only a 3\% increase in JJ count. 

\begin{table}[htbp]
\centering
\caption{Placement results using layout-aware depth optimized netlists compared to placement results using JJ-minimized netlists.}
\label{tab:placement_results}
\scriptsize
\setlength{\tabcolsep}{3pt}
\renewcommand{\arraystretch}{1.08}
\begin{adjustbox}{max width=\columnwidth}
\begin{tabular}{l ccc ccc cc}
\toprule
\multirow{3}{*}{\textbf{Benchmark}} &
\multicolumn{8}{c}{\textbf{Placed Results~\cite{AQFPPlacement}}}\\
\cmidrule(lr){2-9}
&
\multicolumn{3}{c}{\textbf{JJ-Minimized}} &
\multicolumn{3}{c}{\textbf{Layout-Aware}} &
\multicolumn{2}{c}{\multirow{2}{*}{\shortstack{\textbf{Layout-Aware}\\[-1pt]\textbf{/ JJ-Minimized}}}}\\
&
\multicolumn{3}{c}{\scriptsize TCAD'25~\cite{BS_TCAD_Lee}} &
\multicolumn{3}{c}{\scriptsize This Work} &
&
\\
\cmidrule(lr){2-4}\cmidrule(lr){5-7}\cmidrule(lr){8-9}
&
\textbf{Depth} &
\textbf{JJs} &
\textbf{Area ($mm^2$)} &
\textbf{Depth} &
\textbf{JJs} &
\textbf{Area ($mm^2$)} &
\textbf{JJs} &
\textbf{Area}\\
\midrule
\texttt{adder1}    & 8   & 74      & 0.120  & 8   & 78      & 0.111  & 1.05 & 0.93 \\
\texttt{adder8}    & 37  & 1,400   & 3.68  & 37  & 1,642   & 2.22  & 1.17 & 0.60 \\
\texttt{mult8}     & 123 & 11,520  & 29.42 & 128 & 13,194  & 19.55 & 1.15 & 0.66 \\
\texttt{counter16} & 20  & 362      & 0.99  & 20  & 416      & 0.82  & 1.15 & 0.83 \\
\texttt{counter32} & 25  & 900      & 2.46  & 28  & 1,058   & 2.06  & 1.18 & 0.84 \\
\texttt{counter64} & 47  & 2,846   & 9.14  & 46  & 3,128   & 6.85  & 1.10 & 0.75 \\
\texttt{counter128}& 92  & 10,008  & 35.56 & 80  & 9,516   & 23.78 & 0.95 & 0.67 \\
\texttt{c17}       & 5   & 60       & 0.112  & 5   & 62       & 0.106  & 1.03 & 0.95 \\
\texttt{c432}      & 56  & 3,530   & 6.45  & 52  & 3,434   & 5.21  & 0.97 & 0.81 \\
\texttt{c499}      & 74  & 11,386  & 28.05 & 80  & 12,098  & 20.79 & 1.06 & 0.74 \\
\texttt{c880}      & 94  & 12,130  & 24.03 & 77  & 10,192  & 13.74 & 0.84 & 0.57 \\
\texttt{c1355}     & 89  & 14,288  & 37.18 & 84  & 12,354  & 18.53 & 0.86 & 0.50 \\
\texttt{c1908}     & 77  & 8,936   & 18.39 & 70  & 8,230   & 12.00 & 0.92 & 0.65 \\
\texttt{c2670}     & 99  & 21,488  & 39.04 & 84  & 20,248  & 28.17 & 0.94 & 0.72 \\
\texttt{c3540}     & 211 & 45,468  & 156.58& 188 & 40,222  & 60.92 & 0.88 & 0.39 \\
\texttt{c5315}     & 276 & 129,022 & 383.37& 299 & 142,468 & 192.34& 1.10 & 0.50 \\
\texttt{c6288}     & 507 & 83,188  & 169.54& 483 & 80,344  & 118.97& 0.97 & 0.70 \\
\texttt{c7552}     & 295 & 126,212 & 282.00& 311 & 150,632 & 212.51& 1.19 & 0.75 \\
\texttt{sorter32}  & 63  & 7,104   & 12.66 & 63  & 7,032   & 10.62 & 0.99 & 0.84 \\
\texttt{sorter48}  & 120 & 20,864  & 43.30 & 121 & 20,934  & 32.95& 1.00 & 0.76 \\
\texttt{alu32}     & 471 & 130,762 & 440.95& 510 & 138,448 & 212.41& 1.06 & 0.48 \\
\midrule
\multicolumn{7}{r}{\textbf{Average}} &
\textbf{1.03} & \textbf{0.70} \\
\bottomrule
\end{tabular}
\end{adjustbox}
\end{table}

\subsubsection{JJ Impacts on Power}
This modest increase in JJ count has negligible impact on total AQFP energy consumption. Each additional Josephson junction contributes approximately 5~zJ of energy~\cite{cooling_overhead}. Even in the worst case, the additional energy overhead remains negligible. 

The largest increase in power occurs for \texttt{c7552}, which achieves a 69.5~mm$^2$ reduction in area at the cost of 24,420 additional JJs, corresponding to only 0.61~$\mu$W of additional power at 5~GHz. Conversely, the largest area reduction is obtained for \texttt{alu32}, where layout-aware optimization reduces area by 228.5~mm$^2$ while increasing power at 5~GHz by only 0.19~$\mu$W due to 7,686 additional JJs. These results demonstrate that modest increases in JJ count incur negligible energy overhead while enabling substantial reductions in physical area.

\subsection{Runtime Scalability}\label{sec:scalablity}

Although the proposed framework optimizes a more complex objective than prior AQFP legalization approaches, its runtime remains practical for both current and future AQFP designs. The largest fabricated AQFP circuits contain approximately 20,000 JJs~\cite{MANA,JJ20K}, whereas the largest benchmarks evaluated in Tables~\ref{tab:layout_depth_compare_runtime} \&~\ref{tab:placement_results} exceed current fabrication limits by more than $6\times$. Across these benchmarks, the longest runtime is only 6.11~s for \texttt{alu32}, which remains negligible compared to the subsequent physical design stage, whose runtime typically exceeds one hour for this circuit~\cite{AQFPPlacement}.

To evaluate scalability on substantially larger circuits, Table~\ref{tab:runtime_scalability} reports runtimes for representative EPFL benchmarks ranging from 5,000 to nearly 48,000 logic gates. For minimum-depth targeted optimization, legalization completes in less than three minutes even for the largest design. Allowing depth optimization increases the maximum runtime to only 318.2~s for the 47,662-gate \texttt{mem\_ctrl} benchmark.

These results demonstrate that layout-aware AQFP legalization scales efficiently to circuit sizes well beyond those currently manufacturable. In particular, the largest benchmark, \texttt{mem\_ctrl}, contained 921,434 Josephson junctions (JJs), over $46\times$ larger than the largest AQFP circuits fabricated to date~\cite{MANA,JJ20K}. This indicates that the proposed framework can support substantial future increases in AQFP integration density without introducing a significant computational bottleneck into the design flow.

\begin{table}[htbp]
\centering
\caption{Runtime results across large circuits. Runtime is reported in seconds.}
\label{tab:runtime_scalability}
\scriptsize
\setlength{\tabcolsep}{3pt}
\renewcommand{\arraystretch}{1.08}
\begin{adjustbox}{max width=\columnwidth}
\begin{tabular}{l c ccc ccc}
\toprule
\multirow{2}{*}{\textbf{Benchmark}} &
\multirow{2}{*}{\textbf{Gates}} &
\multicolumn{3}{c}{\textbf{Minimum Depth}} &
\multicolumn{3}{c}{\textbf{Depth Optimized}} \\
\cmidrule(lr){3-5} \cmidrule(lr){6-8}
&
&
\textbf{Width ($\mu m$)} &
\textbf{Depth} &
\textbf{Runtime (s)} &
\textbf{Width ($\mu m$)} &
\textbf{Depth} &
\textbf{Runtime (s)} \\
\midrule
\texttt{sin}         & 5,218  & 12,260  & 226 & 1.2   & 10,840  & 233 & 7.6   \\
\texttt{arbiter}     & 12,204 & 27,660  & 90  & 0.5   & 23,820  & 91  & 22.5  \\
\texttt{voter}       & 13,922 & 42,320  & 113 & 7.9   & 42,160  & 114 & 11.2  \\
\texttt{square}      & 17,523 & 23,320  & 257 & 31.5  & 22,300  & 261 & 45.7  \\
\texttt{multiplier}  & 24,817 & 18,340  & 282 & 45.2  & 18,320  & 283 & 57.2  \\
\texttt{log2}        & 29,828 & 33,380  & 422 & 33.8  & 33,060  & 424 & 47.0  \\
\texttt{mem\_ctrl}   & 47,662 & 175,700 & 169 & 164.7 & 160,080 & 170 & 318.2 \\
\bottomrule
\end{tabular}
\end{adjustbox}
\end{table}

\subsection{Flow Integration Validation}

To validate compatibility with existing AQFP design flows, the proposed layout-aware technology legalization has been integrated into the complete AQFP RTL-to-GDSII EDA flow provided by qPALACE~\cite{qPALACE}. Following technology legalization, all benchmark circuits undergo placement, routing, timing verification, and GDSII generation, with functional correctness verified through both Verilog simulation and JoSIM superconducting circuit simulation.

Figure~\ref{fig:GDS} presents the resulting GDSII layouts for \texttt{adder8} generated using the proposed layout-aware legalization and the JJ-minimization legalization of~\cite{BS_TCAD_Lee}. The layouts visually illustrate the motivation behind the proposed optimization. The JJ-minimized design exhibits a highly concentrated placement row that determines the overall die width, increasing total area despite the downstream logic remaining sparsely populated. In contrast, the proposed legalization distributes logic more uniformly across placement rows, reducing peak row width and producing a more compact physical implementation. Although this redistribution incurs a moderate increase of 242 Josephson junctions, it reduces the die size from $2\,\text{mm} \times 4.3\,\text{mm}$ to $1.3\,\text{mm} \times 4\,\text{mm}$, corresponding to a 40\% reduction in die area. Both layouts successfully achieved timing closure at 5 GHz. These results confirm that optimizing the legalization objective for physical layout, rather than junction count alone, translates directly into improved area throughout the complete RTL-to-GDSII design flow.

\begin{figure}[htbp]
\includegraphics[width=0.95\columnwidth]{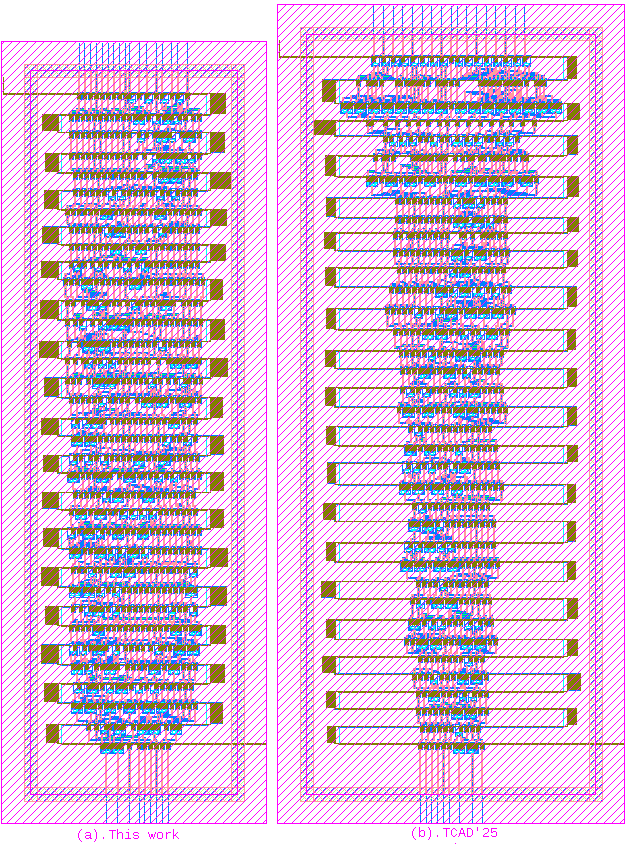}% This is a *.png file
%\vspace{-0.6cm}
\centering
\caption{Final \texttt{adder8} designs produced using layout-aware optimization ($1.3\,\text{mm} \times 4\,\text{mm}$) vs JJ minimization~\cite{BS_TCAD_Lee} ($2\,\text{mm} \times 4.3\,\text{mm}$) as the technology legalization step.}\label{fig:GDS}
\end{figure}

\begin{figure*}[htbp]
\includegraphics[width=0.95\linewidth]{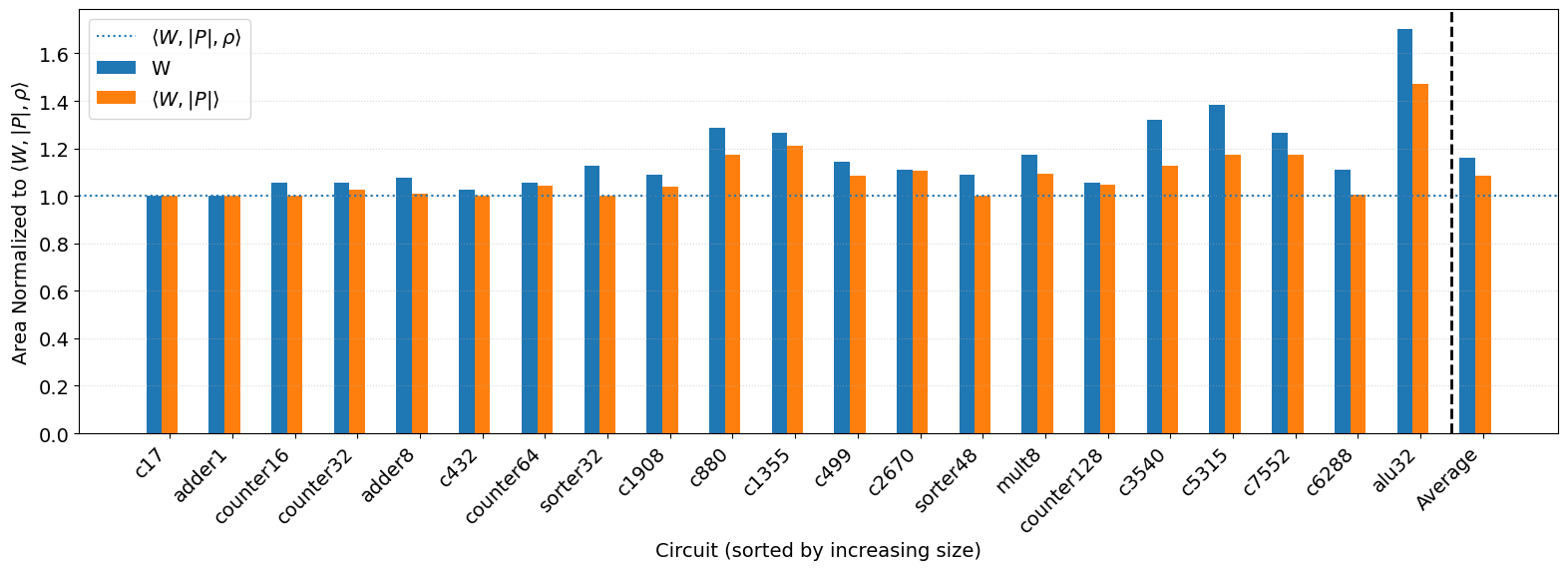}% This is a *.png file
%\vspace{-0.6cm}
\centering
\caption{Comparison of layout-aware results using a cost function of just $\mathcal{W}$ or $\langle\mathcal{W}, |\mathrm{P}|\rangle$ normalized against results using our full cost function of $\langle\mathcal{W}, |\mathrm{P}|, \rho\rangle$ .}\label{fig:cost_ablation}
\end{figure*}

\begin{figure}[htbp]
\includegraphics[width=0.95\linewidth]{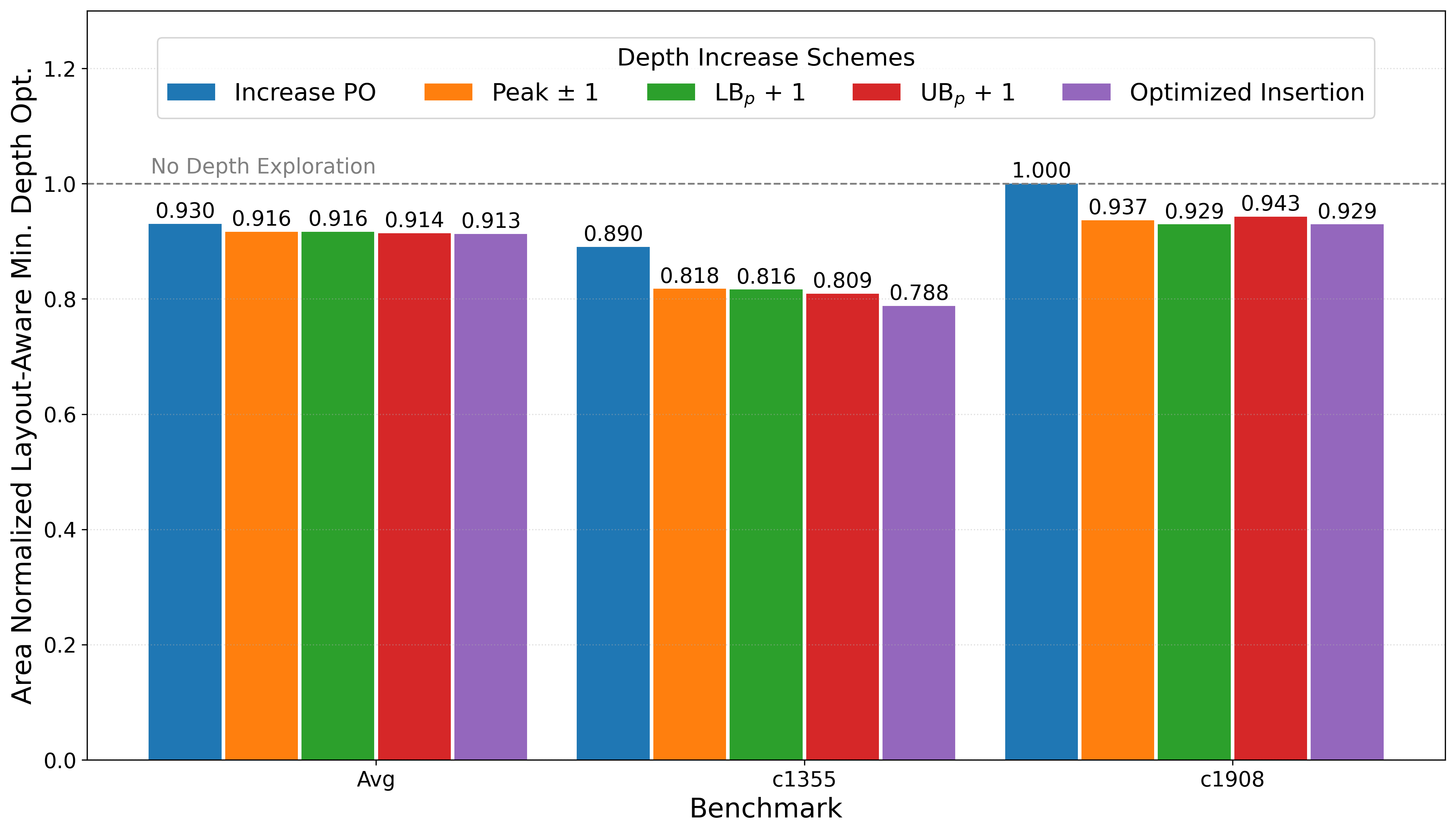}% This is a *.png file
%\vspace{-0.6cm}
\centering
\caption{Ablation evaluating impact of the location where the empty row is inserted during depth exploration.}\label{fig:depth_increase}
\end{figure}

\begin{figure}[htbp]
\includegraphics[width=0.95\linewidth]{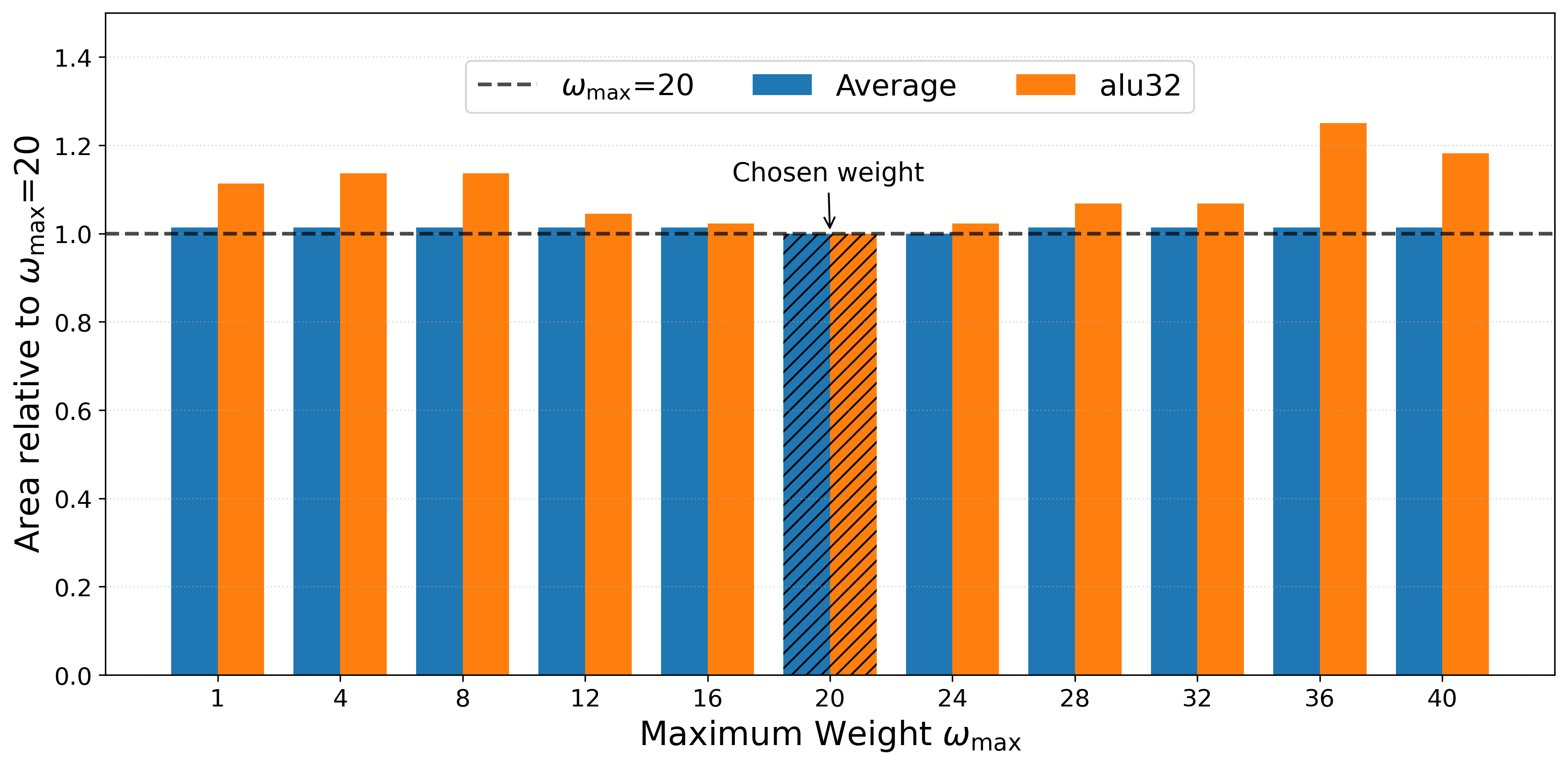}% This is a *.png file
%\vspace{-0.6cm}
\centering
\caption{Ablation studying the impact of the chosen maximum weight of a peak row.}\label{fig:weight_ablation}
\end{figure}

\subsection{Ablations}

The proposed optimization framework combines several interacting ideas to achieve substantial reductions in physical area. To understand the role of each component, we perform a series of ablation studies that separately evaluate the proposed cost function, heuristic depth exploration, and critical-density weighting. These experiments provide insight into which components drive the observed improvements and demonstrate the robustness of the proposed framework.

Figure~\ref{fig:cost_ablation} evaluates the contribution of each component of the proposed lexicographic objective. Optimizing circuit width alone produces the poorest results, while incorporating the number of peak rows provides a consistent improvement across nearly all benchmarks. Adding the critical-density metric yields an additional reduction in area, particularly for larger circuits where congestion surrounding peak rows is more pronounced. These results demonstrate that the secondary objectives effectively differentiate candidate moves within optimization plateaus, steering legalization toward solutions that reduce surrounding congestion rather than merely redistributing it. Consequently, the proposed cost hierarchy of $\langle\mathcal{W}, |\mathrm{P}|, \rho\rangle$ provides progressively richer guidance, with each successive metric resolving ambiguities left by the preceding objectives.

Figure~\ref{fig:depth_increase} evaluates several strategies for selecting the location of the inserted empty row during depth exploration. All insertion schemes achieve significant improvements over fixed-depth optimization, demonstrating that the additional placement flexibility provided by increased depth is the dominant factor. Targeting the critical region surrounding peak rows provides a consistent, albeit modest, improvement over naive insertion schemes, with the proposed heuristic achieving the best overall performance.

Figure~\ref{fig:weight_ablation} evaluates the impact of the critical-density weighting parameter $\omega_{\max}$. The proposed framework exhibits strong robustness to the choice of $\omega_{\max}$, with average area varying by less than a few percent across the evaluated range. While the average behavior is largely insensitive to the selected weight, larger congestion-dominated benchmarks such as \texttt{alu32} benefit from the chosen value of $\omega_{\max}=20$, demonstrating the effectiveness of emphasizing critical density in the vicinity of peak rows.  

Taken together, these studies show that the proposed optimization objective is the dominant source of improvement, while the remaining heuristics provide complementary gains and exhibit strong robustness across benchmark circuits.

\section{Conclusions}

This work defines a layout-aware optimization problem within AQFP technology legalization by introducing the circuit width--depth product as an optimization objective that more accurately captures final physical design area than conventional JJ-count minimization. We proved that the resulting optimization problem is NP-complete and developed heuristic algorithms that efficiently optimize this objective while scaling to large benchmark circuits.

Experimental results demonstrated that the proposed framework consistently improves physical design quality, achieving an average post-placement area reduction of 30\% compared to the state-of-the-art, with more than $2\times$ area reduction on several large benchmark circuits. Although post-synthesis optimization initially introduces an average 17\% increase in JJ count, the improved row distribution substantially simplifies physical implementation, reducing the final overhead to only 3\% after placement while consistently producing smaller layouts.

Beyond the immediate area improvements, this work demonstrates placement-aware objectives can substantially improve physical design costs in AQFP. We hope our results motivate future AQFP design flows that jointly optimize synthesis and physical layout to better capture downstream implementation costs. Future work includes incorporating more sophisticated global search techniques to escape local minima, integrating technology mapping with layout-aware logic distribution, and developing legalization objectives that more tightly model downstream placement and routing challenges.

\section{Acknowledgments}
We used OpenAI ChatGPT and Google Gemini only for grammar correction and language polishing. These tools did not generate research ideas, technical content, results, figures,
or references.

\bibliographystyle{IEEEtran}
\bibliography{bibliography}

\end{document}